\documentclass[onecolumn,review]{elsarticle}
\usepackage{epsfig}  
\usepackage{breakcites,setspace,amssymb,ulem,soul}
\newcommand*{\temp}{\multicolumn{1}{r|}{}} 
\usepackage{epstopdf,xcolor,epsfig,mathdots}
\usepackage{hhline}
\usepackage{graphicx}
\usepackage[utf8]{inputenc}
\usepackage[english]{babel}
\usepackage[cal=boondoxo]{mathalfa}
\usepackage[T1]{fontenc}		
\usepackage{longtable,breakcites}
\usepackage{amssymb,amsmath,amsthm,amsbsy,graphicx}
\usepackage{amsmath,amssymb,amsfonts,enumitem}
\newcounter{defcounter}
\newtheorem{thrm}{Theorem}[section]
\theoremstyle{definition}
\newtheorem{defn}{Definition}[section]
\newtheorem{fact}{Fact}[section]
\usepackage{chngcntr}
\usepackage{apptools}
\AtAppendix{\counterwithin{fact}{subsection}}
\newtheorem{lemma}{Lemma}[section]
\theoremstyle{definition}
\newtheorem{remark}{Remark}[section]
\newtheorem{corollary}{Corollary}[section]
\newtheorem{proposition}{Proposition}[section]
\theoremstyle{definition}
\newtheorem{example}{Example}[section]
\usepackage{hyperref}
\hypersetup{
     pdfmenubar=true,        
     pdffitwindow=true,     
     pdfauthor={Sujay Kadam},     
     pdfnewwindow=true,      
     colorlinks=true,
     citecolor=blue!50!green,
     filecolor=blue,
     urlcolor=blue!50!cyan,
     linkcolor=blue!50!green
}
\journal{European Journal of Control}
\begin{document}
\begin{frontmatter}
\title{{Trackability for Discrete-Time LTI Systems: A Brief Review and New Insights}
\tnoteref{mytitlenote}}
\tnotetext[mytitlenote]{This work was supported by Department of Science and Technology through projects SR/S3/MERC/0064/2012 and SERC/ET-0150/2012, and in part by IIT Gandhinagar.}
\cortext[sdk1]{Corresponding author}
\author{Sujay~D.~Kadam\corref{sdk1}\fnref{sdk} }
\ead{kadam.sujay@iitgn.ac.in}
\address{SysIDEA Lab, Indian Institute of Technology Gandhinagar, Palaj, Gandhinagar, Gujarat, India.}

\author{Harish J.~Palanthandalam-Madapusi  \corref{}\fnref{hpm}}
\ead{harish@iitgn.ac.in}

\fntext[sdk]{Doctoral student in Electrical Engineering at the Indian Institute of Technology, Gandhinagar.}
\fntext[hpm]{Associate Professor in Mechanical Engineering at the Indian Institute of Technology, Gandhinagar.}
\address{SysIDEA Lab, Indian Institute of Technology Gandhinagar, Palaj, Gandhinagar, Gujarat, India.}

\begin{abstract}
`Trackability', the ability of systems to follow arbitrary reference commands, is investigated in this work.
Controllability is not useful in explaining the tracking behavior of system outputs, a gap that is often overlooked.
Trackability addresses this gap by characterizing whether system dynamics permits tracking of certain reference commands, irrespective of the controller/control action used.
While earlier literature has discussed closely-related ideas, lack of consistency in terminology and focus in the literature has necessitated defining trackability in a simple but rigorous manner.
We present definitions for trackability based on elementary linear algebra and propose a rank test for determining trackability of systems, while posing tracking control as an existence question.
We discuss in some detail the relationship of this rank test with  other results in the literature.
These results are then used to generate insights about trackable systems and more importantly to gain insights in to tracking behavior of systems that are not trackable.
We also define three indices that indicate the expected tracking behavior of the system.
Furthermore, we also present a Venn diagram explaining in detail the connections between trackability and other fundamental system properties like controllability, observability and output controllability, while discussing several facts that elaborate these connections.
Specifically, we show that trackability is fundamentally different from controllability and output controllability.
The presented work is expected to serve as a foundation for a deeper investigation into the topic of tracking control and provide a framework for receiving new theoretical insights on trackability.
\end{abstract}
\end{frontmatter} 
\section{Introduction}
The concept of controllability is a fundamental concept in control and has been discussed in the literature and put to use in applications for the past several decades.
However, controllability is not useful in explaining or determining whether it is possible for a system to track a specific reference command \cite{kreindlersarachik1964OC}.
This gap in understanding is often overlooked by the average controls engineer.
There are many examples that help in highlighting this gap, but as a trivial example, one can see that a system with more outputs than inputs will not be able to track independently all the reference outputs regardless of the controller employed even if the system is fully controllable (or more precisely, state controllable).
The same can be said about output controllability as well.

Therefore, there is a need to define an underlying system property that characterizes the ability of a system to track arbitrary reference commands or a way to characterize whether the system dynamics allow for certain reference trajectories to be followed (even before designing a controller).
Note that this is an existence question, and not a question related to actual tracking performance with a specific controller. 
We refer to this ability of a system to track arbitrary reference commands as `Trackability'.

Earlier literature has discussed ideas related to trackability in some detail \cite{kreindlersarachik1964OC,brockett1965reproducibility,dorato1969inverse,sain1969invertibility,albrecht1981reproducibility,brockett2015finite,kurek2002trackability,garcia2013alternative}. However, the terminology and focus seem to be varying in these works. These works cover several closely-related concepts such as  { functional reproducibility in \cite{brockett1965reproducibility,sain1969invertibility,albrecht1981reproducibility}, servo mechanism sense controllability in \cite{brockett2015finite}, functional output controllability in \cite{garcia2013testing,garcia2013alternative}, dynamic path controllability in  \cite{maas1994dynamic}, target path controllability in \cite{wohltmann1985target,engwerda1988control}, functional output $\epsilon$ - controllability in \cite{germani1983functional}, perfect output controllability in \cite{aoki}, right-invertibility in \cite{nijmeijer1986right}, and trackability \cite{kurek2002trackability,sandoz2008trackability}.}
This has led to a lack of consistency and common agreement on definitions.
Ideas related to trackability have also been discussed in works related to synthesis of tracking controllers such as  inversion based feedforward tracking control methods \cite{dev2,dev3,devasia1996nonlinear,george1999stable,leang2009feedforward} using either closed-loop inversion approach or the plant inversion approach \cite{clayton2009review}, and more recently, tracking control through signal decomposition \cite{kasemsinsup2017reference}.
Recently, a review of inversion-based approaches is presented in \cite{van2018inversion}.

{In this paper, we argue that there is a need for defining trackability in simple and rigorous terms. 
It may be noted that `functional reproducibility' due to Brockett and Mesarovi\'c \cite{brockett1965reproducibility} and the notion of trackability presented here are analogous, functional reproducibility concerns itself with continuous-time whereas trackability concerns itself with discrete-time systems. Given the natural association of the term trackability with tracking control, we chose to adopt this new terminology rather than use `functional reproducibility'.  While the idea is not entirely new, the approach taken and the framework used is intended to give the reader new insights into this property of trackability and explain its importance by putting it together with other fundamental system properties.  
This approach also facilitates development of new tools that may be useful for tracking control problems.}

The key contributions of this paper are multifold.
First, we present the notion of trackability for discrete-time LTI systems based on definitions of trackable and untrackable sets, and propose a rank test for determining trackability of a system in Section \ref{sec:trackability}.
Secondly, these results are then put in the context of several different results in the literature thus also serving as a brief review of various concepts in the literature.
The third contribution is in terms of results that are used to generate insights about tracking behavior of systems that are not trackable, in Section \ref{sec:untrk}.
We also present in Section \ref{sec:indices}, three trackability indices that indicate the expected tracking behavior of a system in general or in response to a specific reference command, which is our fourth contribution.
These indices are applicable to both trackable and untrackable systems, and make use of system's Markov parameters to quantify expected tracking behavior.
We discuss elaborately the connections between trackability and other fundamental system properties like controllability and observability in Section \ref{sec:venn}, followed by several related remarks, which constitutes our fifth contribution. 
The discussion in Section \ref{sec:venn} is intended to provide the reader with a solid overview of various system properties and their interrelations. Section \ref{sec:conclusion} draws conclusions about the presented work.
Additionally, in the appendix, the paper also presents a number of examples to illustrate 17 different possibilities of combinations of system properties to further clarify the results in the paper.

\section{Trackability Results}\label{sec:trackability}
 \subsection{Preliminaries and Definitions}
Consider the discrete-time LTI system represented by the state space equations
\begin{equation}
\label{SS1}
x_{k+1}=A x_k + B u_k,
\end{equation}
\begin{equation}
\label{SS2}
y_{k}=C x_k,
\end{equation}
where $A \in \mathbb{R}^{n \times n}$, $B \in \mathbb{R}^{n \times m}$, $C \in \mathbb{R}^{l \times n}$,  and $x_k \in \mathbb{R}^{n}$, $u_k \in \mathbb{R}^{m}$ and $y_k \in \mathbb{R}^{l}$ are the states, inputs and the outputs respectively. 
{The non-negative integer $k \in \mathbb{N} \cup \lbrace 0 \rbrace$  represents the sample time index. We also assume $m,l \leq n$}. 

{Let this system \eqref{SS1}, \eqref{SS2} with an initial condition $x_0=0$ be applied  a unit impulse signal at $k=0$, having its $q^{\rm th}$ component $1$ and all other components $0$, such that,
\begin{equation}\label{impulse}
u_k =\begin{cases}
\delta_q \text{ for } k =0,\\
0, \text{ otherwise,}
\end{cases}
\end{equation}
where the $j^{\rm th}$ component of $\delta_q \in \mathbb{R}^m$ is
\begin{equation}
\delta_{q,j} =\begin{cases}
1 \text{ for } j=q,\\
0, \text{ otherwise.}
\end{cases}
\end{equation}
Then, the $m$ positive integers $L_q$ are defined such that corresponding to each $1 \leq q \leq m$ such that, each $L_q$ is the smallest $k$ for which $y_k \neq 0$, in response to unit impulse signal in \eqref{impulse}. In other words, each $L_q$ represents the delay between the unit impulse signal applied at $k=0$, exclusively on the $q^{\rm th}$ component of the input vector and its effect seen on the output vector.
A positive integer $L \in \mathbb{N}$ may now be defined as \begin{equation}
L \triangleq \operatorname{min}\left( L_1, L_2,\cdots,L_m \right).\end{equation}
The positive integer $L$ therefore represents the smallest possible delay between any of the input-output pairs for the system \eqref{SS1}, \eqref{SS2}.
The facts that follow are now stated following the definition of $L$.
\begin{remark}
Systems for which $L$ is not defined, have no input-output relationship and therefore appear as a transfer function with a static gain zero. Hence, for the remainder of the paper, we consider systems with a defined $L$.
\end{remark}
\begin{fact}
For every system \eqref{SS1}, \eqref{SS2}, either $1 \leq L \leq n$ or $L$ does not exist.
\end{fact}
\begin{fact}
For every system \eqref{SS1}, \eqref{SS2}, $CA^{L-1}B \neq 0$, and $CA^{i-1}B = 0$, for $1\leq i \leq L-1$. 
\end{fact}}

{\begin{remark}
Since $L$ determines the smallest time within which an input may influence the output vector, it is logical to specify reference signals starting at sample number $L$ instead of $0$.
\end{remark}}

{For $r \geq L$ number of time samples, we define respectively, the sequences $\mathcal{U}_r \in \mathbb{R}^{(r+1)m}$ and {$\mathcal{Y}_{r,q} \in \mathbb{R}^{(r-q+1)l}$} as
\begin{equation}\label{YrUr}
\mathcal{U}_{r}\triangleq
\begin{bmatrix}
u_0\\u_1\\ \vdots \\u_{r}
\end{bmatrix}~\text{and}~
{\mathcal{Y}_{r,q}~\triangleq~
\begin{bmatrix}
y_q\\y_{q+1}\\ \vdots \\y_r
\end{bmatrix}}
\end{equation}
representing the system inputs and outputs respectively. Note that, the first entry of the output sequence {$\mathcal{Y}_{r,q}$ is $y_q$, which enables systems with delay to be taken into account, by choosing $q=L$.} 
{We also define matrices $\Gamma_{r} \in \mathbb{R}^{(r-L+1)l \times n}$ and $M_{r} \in \mathbb{R}^{(r-L+1)l \times rm}$ dependent on system parameters $A,B,C$ and positive integers $r$ and $q$ as
\begin{equation}\label{GrMr}
\Gamma_{r} \triangleq \begin{bmatrix}
CA^L\\
CA^{L+1}\\
\vdots \\
CA^{r}
\end{bmatrix}~\text{and}~M_{r} \triangleq \begin{bmatrix}
CA^{L-1}B &0  &\cdots &0 \\
CA^{L}B &CA^{L-1}B &\cdots &0 \\
\vdots &\vdots  &\ddots &0 \\
CA^{r-L}B &CA^{r-L-1}B &\cdots &CA^{L-1}B\\
\end{bmatrix}.
\end{equation}}

The Markov parameter matrix $M_{r}$ can be seen as a convolution matrix of the impulse response of the system determined for $r$ samples, { starting from $L^{\rm th}$ sample.
It is important to note that for a system \eqref{SS1}, \eqref{SS2} with a delay $L$, the matrices $CA^i$, need not be zero, even if $CA^iB=0$ for $0 \leq i <L$ .}

Batch equations for the recursive state space equations (\ref{SS1}) and (\ref{SS2}) {for a system with delay $L$,} initialized at initial condition $x_0$ may then be written as 
\begin{equation}\label{bigeq}
\mathcal{Y}_{r,L}= \Gamma_{r} x_0 + M_{r} \mathcal{U}_{r-L}.
\end{equation}

We first define the notions of trackable set and untrackable set, using which the system's property of trackability is defined. {By $\mathcal{Y}_{{\rm ref,}{r,L}}$, we denote
a reference command sequence for $r-L+1$ time samples starting from $L$ through $r$.}
{\begin{defn} \label{TrSSdef}
The \underline{trackable set} of the system (\ref{SS1}), (\ref{SS2}) {with a delay $L$}, for $x_0 \in \mathbb{R}^n$, and at time index $r$ is the set $\mathfrak{T}_{r,L} (A,B,C,x_0)$  given by
{\begin{equation}
\begin{split}
\mathfrak{T}_{r,L} (A,B,C,x_0) \triangleq \lbrace \mathcal{Y}_{{\rm ref,}{r,L}} \in \mathbb{R}^{(r-L+1)l}:
\text{there exists~} \mathcal{U}_{r-L}\text{~that satisfies }\\
{{\mathcal{Y}}_{{\rm ref,}{r,L}}}-\Gamma_{r} x_0 = M_{r} \mathcal{U}_{r-L} \rbrace,
\end{split}
\end{equation}}
and the corresponding \underline{untrackable set} is the set given by \begin{equation} \mathfrak{U}_{r,L} (A,B,C,x_0)\triangleq\mathbb{R}^{(r-L+1)l}-\mathfrak{T}_{r,L} (A,B,C,x_0). \end{equation} The untrackable set is the complement of the trackable set, and is the collection of reference commands that the system cannot follow. 
\end{defn}}
\subsection{Facts about initial conditions and trackability sets}
We recall facts from \cite{kadam2017revisiting} relating to the initial conditions and the trackable and untrackable sets of the system. Note that $\mathcal{R}(M_{r})$ and $\mathcal{N}(M_{r})$ are used to denote the column space and null space of $M_{r}$  respectively.
\begin{fact}
\label{fact_4}
For every $x_0 \in \mathbb{R}^n$, $\Gamma_{r} x_0 \in \mathfrak{T}_{r,L}(A,B,C,x_0)$ and $ \mathfrak{T}_{r,L}(A,B,C,x_0)=\mathcal{R}(M_{r}) + \lbrace{\Gamma_{r} x_0}\rbrace$.
\end{fact}

\begin{fact} \label{fact_1}
Let $\mathcal{R}(\Gamma_{r}) \subseteq \mathcal{R}(M_{r})$. Then, $\mathfrak{T}_{r,L}(A,B,C,x_0)$ is independent of $x_0$ and $\mathfrak{T}_{r,L}(A,B,C,x_0)=\mathcal{R}(M_{r})$.
\end{fact}

Since trackable set is independent of initial conditions in the situation discussed in Fact \ref{fact_1}, we may denote, $\mathfrak{T}_{r,L}(A,B,C,x_0)$  as $\mathfrak{T}_{r,L}(A,B,C)$ and similarly, $\mathfrak{U}_{r,L} (A,B,C,x_0)$ as $\mathfrak{U}_{r,L} (A,B,C)$.

\begin{fact}\label{fact_2}
Let $\mathcal{R}(\Gamma_{r})~  {\not\subseteq}  ~\mathcal{R}(M_{r})$. Then,   $\mathfrak{T}_{r,L}(A,B,C,x_0) \neq \mathbb{R}^{(r-L+1)l}$,
 and $\mathfrak{U}_{r,L} (A,B,C,x_0)$ is not empty.
\end{fact}

\begin{fact}\label{fact3}
Let $\mathfrak{T}_{r,L}(A,B,C,x_0) = \mathbb{R}^{(r-L+1)l}$ for a given $x_0 \in \mathbb{R}^n$. Then, for all $\bar{x}_0 \in \mathbb{R}^n$, $\mathfrak{T}_{r,L} (A,B,C,\bar{x}_0)=\mathfrak{T}_{r,L} (A,B,C,{x}_0)= \mathfrak{T}_{r,L}(A,B,C)=\mathbb{R}^{(r-L+1)l}$ and $\mathcal{R}(\Gamma_{r}) \subseteq \mathcal{R}(M_{r})$.
\end{fact}

\begin{fact}
Let $x_0=0$ . Then, $\mathfrak{T}_{r,L}(A,B,C,x_0)=\mathcal{R}(M_{r})$.
\end{fact}
\begin{defn}\label{Trdef}
The system (\ref{SS1}), (\ref{SS2}) {with a delay $L$} is \underline{trackable} if $\mathfrak{T}_{r,L}(A,B,C)=\mathbb{R}^{(r-L+1)l}$ for all ${r \geq L}$.
\end{defn}

From Definition \ref{Trdef}, and Facts \ref{fact3} and \ref{fact_r} it should be evident that, trackability is a system property independent of initial conditions and time step $r$. That is, if the system (\ref{SS1}), (\ref{SS2}) is {trackable}, then it is so, regardless of initial conditions and time step $r$. 
Further, $\mathfrak{T}_{r,L}(A,B,C)$ is a vector space and can be referred to as trackable space. Note that $\mathfrak{T}_{r,L}(A,B,C)$ is a vector space even in the situation outlined in Fact \ref{fact_1}, even if $\mathfrak{T}_{r,L}(A,B,C)\neq \mathbb{R}^{(r-L+1)l}$. However, when a system is not trackable and $\mathcal{R}(\Gamma_{r}) \not \subseteq \mathcal{R}(M_{r})$, the initial conditions become relevant when discussing the trackable set. $\mathfrak{T}_{r,L}(A,B,C,x_0)$ is not a space (unless $x_0$=0) in such situations. We will invoke the notation $\mathfrak{T}_{r,L}(A,B,C,x_0)$ for referring to trackable set of such systems. In the subsection that follows, we give a simple rank test based on Markov parameters of the system to determine whether or not a system is trackable

\subsection{Rank test for trackability}
The relationship between ranks of Markov parameter matrix $M_{r}$ and the first non-zero Markov parameter {$CA^{L-1}B$} is discussed in Lemma \ref{lemmaA1}, which is used subsequently to determine conditions for trackability. The Moore-Penrose generalized inverse \cite{bernstein2009matrix} of a real matrix $Q$ is denoted by $Q^\dagger$.
{ \begin{lemma} \label{lemmaA1} Let $r \geq L$. Then,
${\rm rank}\left( M_{r}\right)={\rm min} \left((r-L+1)l,(r-L+1)m \right)$, if and only if ${\rm rank}\left( CA^{L-1}B\right)={\rm min} \left(l,m \right)$ for the system \eqref{SS1}, \eqref{SS2} with delay $L$. 
\end{lemma}}
\begin{proof}
This admits a proof similar to the proof in \cite{kadam2017revisiting}.
\end{proof} 
\begin{fact} \label{fact_r}If there exists an $r \geq L$ for which $\mathfrak{T}_{r,L}(A,B,C)=\mathbb{R}^{(r-L+1)l}$, then it holds for every $r \geq L$. 
\end{fact}
{\begin{proof}
Let $\mathfrak{T}_{q,L}(A,B,C)=\mathbb{R}^{(q-L+1)l}$ hold true for some $q \geq L$. Then, following Definition \ref{TrSSdef}, for all sequences $\mathcal{Y}_{{\rm ref},q,L} \in \mathbb{R}^{(q-L+1)l}$ and for all $x_0 \in \mathbb{R}^n$, there exists a control input sequence ${\mathcal{U}}_{q-L}$ satisfying \begin{equation}\label{bigeqQs}
\mathcal{Y}_{{\rm ref},q,L}= \Gamma_{q,L} x_0 + M_{q,L} \mathcal{U}_{q-L},
\end{equation}
which implies $\operatorname{rank}(M_{q,L})=(q-L+1)l$ and invoking Lemma \ref{lemmaA1}, it follows that $\operatorname{rank}(CA^{L-1}B)=l$, which further implies that  $\operatorname{rank}(M_{r})=(r-L+1)l$ and $\mathfrak{T}_{r,L}(A,B,C)=\mathbb{R}^{(r-L+1)l}$, for all $r \geq L$, therefore proving the fact.
\end{proof}}

For determining a system's trackability, a rank test is proposed in the following theorem.

{\begin{thrm} \label{TrackThm1} The following statements are equivalent:
\begin{enumerate} [label=\roman*)]
\item \label{TT1}System (\ref{SS1}) and (\ref{SS2}) {with a delay $L$} is trackable.
\item \label{TT2}${\rm rank}\left(M_{r}\right)=(r-L+1)l$,  for all $r \geq L$.
\item \label{TT3}${\rm rank}\left(CA^{L-1}B\right)=l$.
\end{enumerate}  
\end{thrm}}
\begin{proof}
To prove \ref{TT1} implies \ref{TT2}, we have $\mathfrak{T}_{r,L}(A,B,C)=\mathfrak{T}_{r,L} (A,B,C,0)=\mathbb{R}^{(r-L+1)l}$, and therefore $\mathcal{R}(M_{r})=\mathbb{R}^{(r-L+1)l}$ and consequently ${\rm rank} (M_{r})=(r-L+1)l$. 

The proof of \ref{TT2} implies \ref{TT3} is immediate from Lemma \ref{lemmaA1}.

To prove \ref{TT3}  implies \ref{TT1}, since ${\rm rank}(CA^{L-1}B)=l$, then $CA^{L-1}B$ is right invertible and an input
\begin{equation}\label{eq:ubar}
{u}_{k} \triangleq (CA^{L-1}B)^{\dagger} \left( y_{{\rm ref},k+L} - \left( CA^{k+L} x_0+\sum _{i=1} ^{k} C A^{i+L-1}B {u}_{k-i}\right)  \right),
\end{equation}
 may be constructed for every $x_0 \in \mathbb{R}^{n}$, $y_{{\rm ref},k}\in \mathbb{R}^l$, $L\leq k\leq r$, such that sequences $\bar{\mathcal{U}}_{r-1}$ and ${\mathcal{Y}}_{{\rm ref},r,L}$ satisfy \eqref{bigeq} for all $\mathcal{Y}_{{\rm ref,}r,L} \in \mathbb{R}^{(r-L+1)l}$. Therefore, $\mathfrak{T}_{r,L}(A,B,C,x_0)=\mathfrak{T}_{r,L}(A,B,C)=\mathbb{R}^{(r-L+1)l}$ and by Definition \ref{Trdef}, system (\ref{SS1}), (\ref{SS2}) with a delay $L$ is trackable.
\end{proof}

{While Theorem \ref{TrackThm1} relates closely to the results in earlier literature that discusses notions similar to trackability, the approach taken in this paper is based on simple but rigorous linear algebraic concepts.}  {The rank condition stated in the theorem only requires the knowledge of one particular Markov parameter to establish trackability and is therefore simple, while also being applicable to a general class of systems having an arbitrary $L$ delay. It is emphasized that trackability is a notion analogous to `functional reproducibility' by Brockett and Mesarovi\'c \cite{brockett1965reproducibility}. }
{To give the interested reader a flavor of developments reported in earlier literature, a brief overview of contributions relating to trackability is presented in Appendix \ref{Erl}.}
\subsection{Connections between trackability and methods for inversion-based control synthesis}
It may be noted that an open loop control input 
\begin{equation}\label{eq:ubar}
{u}_{k} \triangleq (CA^{L-1}B)^{\dagger} \left( y_{{\rm ref},k+L} - \left( CA^{k+L} x_0+\sum _{i=1} ^{k} C A^{i+L-1}B {u}_{k-i}\right)  \right),
\end{equation} may be used to track a trajectory specified by  $y_{{\rm ref},k+L}$ for an system with delay $L$ starting at initial condition $x_0$.  However, this control input may not be of much use in practical situations and a closed loop control action similar to \eqref{eq:ubar} should be defined as
{
\begin{equation}\label{ucir}
{u}_{k}\triangleq(CA^{L-1}B)^\dagger \left( y_{{\rm ref},k+L} - CA^Lx_{k} \right),
\end{equation}
to be useful in practical situations.
 Note that, $\left( CA^{k+L} x_0+\sum _{i=1} ^{k} C A^{i+L-1}B \bar{u}_{k-i}\right)=CA^L x_{k-L}$ and estimates of  $x_{k}$ }can be obtained from output measurements by using a suitable state estimator or a filter, if the system is state observable. 
 
 It is worth noting here that the control scheme in \cite{chavan2015command} uses a control input similar to \eqref{ucir}  with an unbiased minimum variance filter for estimating the states for systems with $L=1$. It may be also be noted that \eqref{ucir} is a special case of the input 
 \begin{equation}\label{uffd}
 u_k =\begin{bmatrix}
c_1 A^{\rho_1}B\\
c_2 A^{\rho_2}B\\
\vdots\\
c_l A^{\rho_l}B
\end{bmatrix}^{-1} \left[\begin{bmatrix}
y_{{\rm ref}_{1,k+\rho_1}}\\
y_{{\rm ref}_{2,k+\rho_2}}\\
\vdots\\
y_{{\rm ref}_{l,k+\rho_l}}
\end{bmatrix} - \begin{bmatrix}
c_1 A^{\rho_1}\\
c_2 A^{\rho_2}\\
\vdots\\
c_l A^{\rho_l}
\end{bmatrix} x_k\right],
 \end{equation}
 given in \cite{george1999stable},  which is a discrete-time analogue of  the control input
\begin{equation} \label{uff}
u_{\rm ff} (t) =\begin{bmatrix}
c_1 A^{\rho_1}B\\
c_2 A^{\rho_2}B\\
\vdots\\
c_l A^{\rho_l}B
\end{bmatrix}^{-1} \left[y_d ^{(\rho)} (t) - \begin{bmatrix}
c_1 A^{\rho_1}\\
c_2 A^{\rho_2}\\
\vdots\\
c_l A^{\rho_l}
\end{bmatrix} x(t)\right],
\end{equation}  for continuous-time systems developed in \cite{zou1999preview}. Note that, $y^{(\rho)}(t) \triangleq \begin{bmatrix}
\frac{{\rm d}^{(\rho_1)}y_1}{{\rm d}t^{(\rho_1)}} & \frac{{\rm d}^{(\rho_2)}y_2}{{\rm d}t^{\rho_2}}& \cdots &\frac{{\rm d}^{(\rho_l)}y_1}{{\rm d}t^{\rho_l}}
\end{bmatrix}
$ and $\rho \triangleq \begin{bmatrix}
{\rho_1} & {\rho_2}& \cdots & {\rho_l}
\end{bmatrix}^{\top} $ is the relative degree. Also, $c_i$ represents the $i^{\rm th}$ row of the $C$ matrix.

It is readily seen that the control inputs \eqref{ucir}, \eqref{uffd}, and \eqref{uff} from the respective inversion-based control synthesis methods are similar to each other, which in turn are similar to \eqref{eq:ubar}. It is also interesting to note that each of these assume right invertibility of a matrix consisting of Markov parameters, which translates to it having full row rank. Theorem \ref{TrackThm1} precisely captures this aspect. Furthermore, this commonality not only establishes a link between inversion-based control synthesis and trackability, but also establishes the adoptability of the idea of trackability across continuous-time and discrete-time systems owing to \eqref{uffd} and \eqref{uff}.
\section{Untrackable Systems and Untrackable Reference Commands}\label{sec:untrk}
In this section, we discuss systems with {${\rm rank}(CA^{L-1}B)< l$}.  These systems can only follow reference command sequences that are in the trackable set $\mathfrak{T}_{r,L}(A,B,C,x_0)$ of the system.  Let us, for simplicity, assume that $x_0 =0$. Then, $\mathfrak{T}_{r,L}(A,B,C,x_0)=\mathcal{R}(M_{r})$ and there exist control inputs for tracking reference commands  $\mathcal{Y}_{{\rm ref,}r,L} \in \mathcal{R}(M_{r})$.
To discuss tracking of reference command sequences that are not in $\mathcal{R}(M_{r})$, we first define the orthogonal projection of $\mathcal{Y}_{{\rm ref,}r,L}$ on $\mathcal{R}(M_{r})$ as
\begin{equation}\label{projeq}
 \mathcal{Y}_{{\rm ref},r,L} \Pi_{\mathcal{R}({M_{r}})} \triangleq M_{r}(M_{r}^{{\top}} M_{r})^{\dagger}M_{r}^{{\top}} \mathcal{Y}_{{\rm ref,}r,L}.
\end{equation}
 A sequence projected on $\mathcal{R}(M_{r})$ lies in $\mathfrak{T}_{r,L} (A,B,C,0)=\mathcal{R}(M_{r})$ and if the system is initialized at $x_0=0$, a control input for exactly tracking this projected sequence exists. We now state some facts about the projected sequence defined in (\ref{projeq}).
\begin{fact} The following statements are true:
\begin{enumerate}
\item $\mathcal{Y}_{{\rm ref,}r,L}\Pi_{\mathcal{R}(M_{r})}$ is orthogonal to $\mathcal{N}(M_{r} ^{\top})$.
\item For any $\mathcal{Y}_{{\rm ref},r,L}$, $\lVert \mathcal{Y}_{{\rm ref},r,L} \rVert_2 \geqslant \lVert \mathcal{Y}_{{\rm ref},r,L} \Pi_{\mathcal{R}(M_{r})} \rVert_2$.
\item $\mathcal{Y}_{{\rm ref,}r,L}-\mathcal{Y}_{{\rm ref,}r,L}\Pi_{\mathcal{R}(M_{r})} \in \mathcal{N}(M_{r} ^{\top})$.
\item $\mathcal{Y}_{{\rm ref,}r,L}\Pi_{\mathcal{R}(M_{r})}$ lies in $\mathfrak{T}_{r,L} (A,B,C,0)=\mathcal{R}(M_{r})$ and has smallest 2-norm distance from $\mathcal{Y}_{{\rm ref,}r,L}$ amongst all trajectories in $\mathfrak{T}_{r,L} (A,B,C,0)=\mathcal{R}(M_{r})$.
\end{enumerate} 
\end{fact}
\begin{proof}
The proofs are trivial once it is noted that $\mathcal{Y}_{{\rm ref},r,L}=\mathcal{Y}_{{\rm ref},r,L} \Pi_{{\mathcal{R}(M_{r})}}+\mathcal{Y}_{{\rm ref},r,L} \Pi_{{\mathcal{N}(M_{r} ^{\top})}}$ and $\mathcal{Y}_{{\rm ref},r,L} \Pi_{{\mathcal{R}(M_{r})}} \perp \mathcal{Y}_{{\rm ref},r,L} \Pi_{{\mathcal{N}(M_{r} ^{\top})}}$.
\end{proof}

\begin{fact}
For the system (\ref{SS1}), (\ref{SS2}), let ${\rm rank}(M_{r})=(r-L+1)l$. Then, $\mathcal{Y}_{{\rm ref,}r,L}\Pi_{\mathcal{R}(M_{r})}=\mathcal{Y}_{{\rm ref,}r,L}$, for all $\mathcal{Y}_{{\rm ref,}r,L}$.
\end{fact}
Note that, the projection defined in (\ref{projeq}) is one projection out of the many possible projections on ${\mathcal{R}(M_{r})}$. 

\begin{proposition}\label{proput}
For every {$\mathcal{Y}_{{\rm ref},r,L} \in \mathbb{R}^{(r-L+1)l}$,} system (\ref{SS1}), (\ref{SS2}) can track the sequence $\mathcal{Y}_{{\rm ref},r,L}\Pi_{\mathcal{R}(M_{r})} \in \mathcal{R}(M_{r})=\mathfrak{T}(A,B,C,0)$ exactly if the system is initialized at $x_0=0$. 
\end{proposition}
Proposition \ref{proput} can elaborated with the help of an example. 
\begin{example} \label{exampleP1}
Consider a system \eqref{SS1}, \eqref{SS2} with matrices \[ 
\begin{split} A=\begin{bmatrix}
         0    &1       & 0       & 0\\
         0         &0    &1         &0\\
         0         &0         &0    &1\\
0.05 &0.1 &0.15 &0.2
\end{bmatrix} ,~B=\begin{bmatrix}
     1     &0\\
     0     &1\\
     0     &8\\
     1     &8
\end{bmatrix}\\ \text{~and~}C=\begin{bmatrix}
     1     &0     &0     &0\\
     0    &0     &1     &0\\
     0    &0     &0     &1
\end{bmatrix} \end{split}\]
having $L=1$.
The system is state controllable as well as output controllable. However, since ${\rm rank}(CA^{L-1}B)=2< l$, the system is untrackable and cannot follow arbitrary reference commands $\mathcal{Y}_{{\rm ref,}r,L} \notin \mathfrak{T}_{r,L}(A,B,C,x_0)$. It can however follow reference commands projected on $\mathcal{R}(M_{r})$ exactly as shown in Fig. \ref{projscl} when initialized at zero initial conditions. In case of non-zero initial conditions, exact tracking cannot be guaranteed, however, in asymptotically stable systems asymptotic tracking is possible.
\end{example}
\begin{proposition}\label{proplCA^{L-1}B}
 Let ${\rm rank}(CA^{L-1}B)=\tilde{l} < l$ and $\tilde{C} \in \mathbb{R}^{\tilde{l} \times n}$ be a matrix containing $\tilde{l}$ rows of $C$ such that ${\rm rank}(\tilde{C}A^{L-1}B)=\tilde{l}$. Then, the system with matrices $A$, $B$, $\tilde{C}$ is trackable.
\end{proposition}
Deleting the third row of $C$ matrix in Example \ref{exampleP1}, we obtain $\tilde{C}=\begin{bmatrix}
     1     &0     &0     &0\\
     0    &0     &1     &0
\end{bmatrix}$ and with $\tilde{C}B=\begin{bmatrix}
     1     &0 \\
     0    &8  
\end{bmatrix}.$ This system is trackable since ${\rm rank}(\tilde{C}B)=\tilde{l}$ and can follow arbitrary reference commands on the chosen $\tilde{l}$ outputs  as shown in Fig. \ref{projscl2}. For the simulation results shown in Fig. \ref{projscl} and Fig. \ref{projscl2}, control input in (\ref{eq:ubar}) was used.
\begin{figure}[!htbp]
\centering
\includegraphics[width=0.8\linewidth]{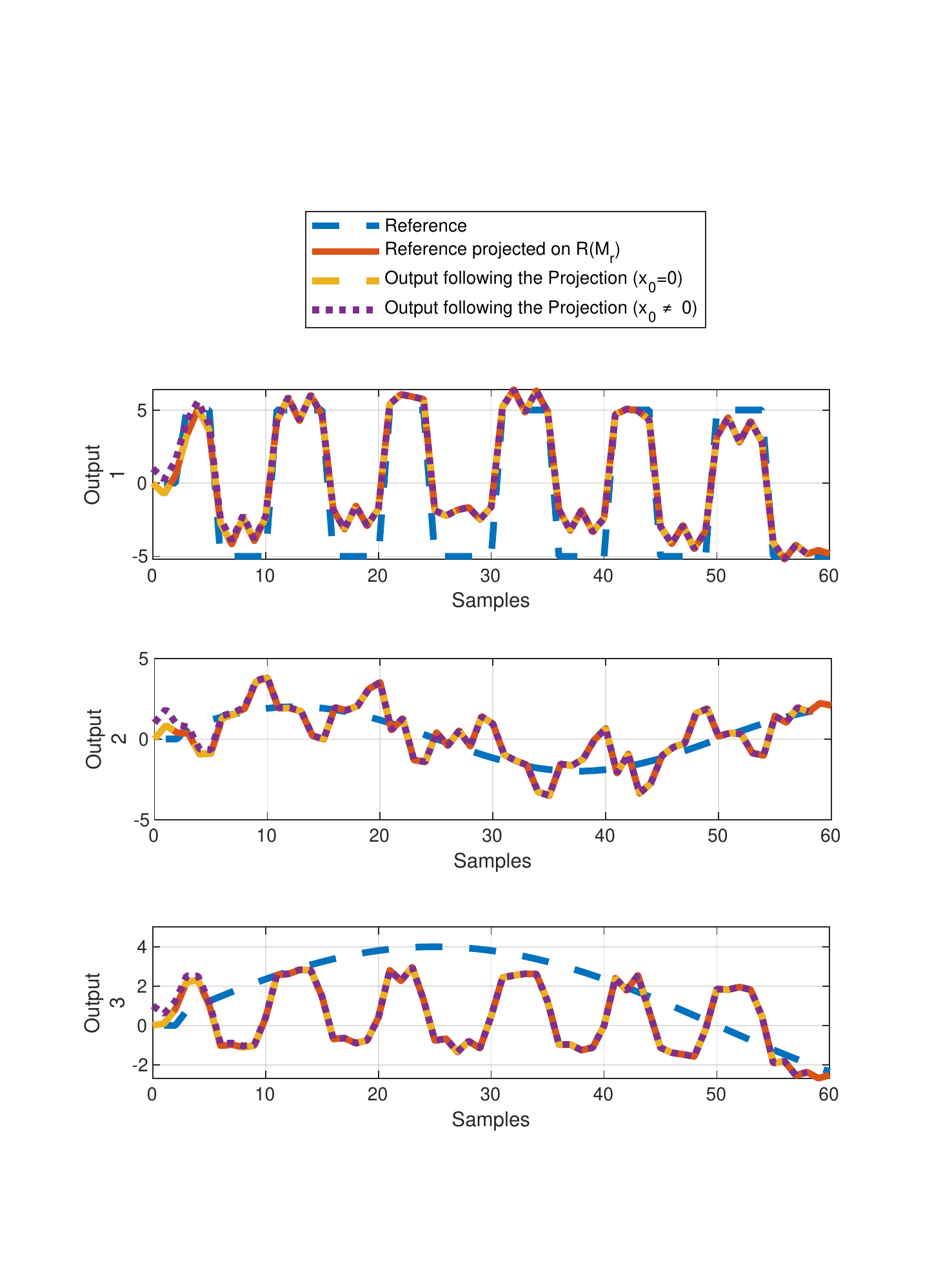}
\caption{Plots showing tracking of references projected on $\mathcal{R}(M_{r})$ when $\mathcal{Y}_{{\rm ref},r,L} \notin \mathcal{R}(M_{r})$. }
\label{projscl}
\end{figure}

\begin{figure}[!htbp]
\centering
\includegraphics[width=0.8\linewidth]{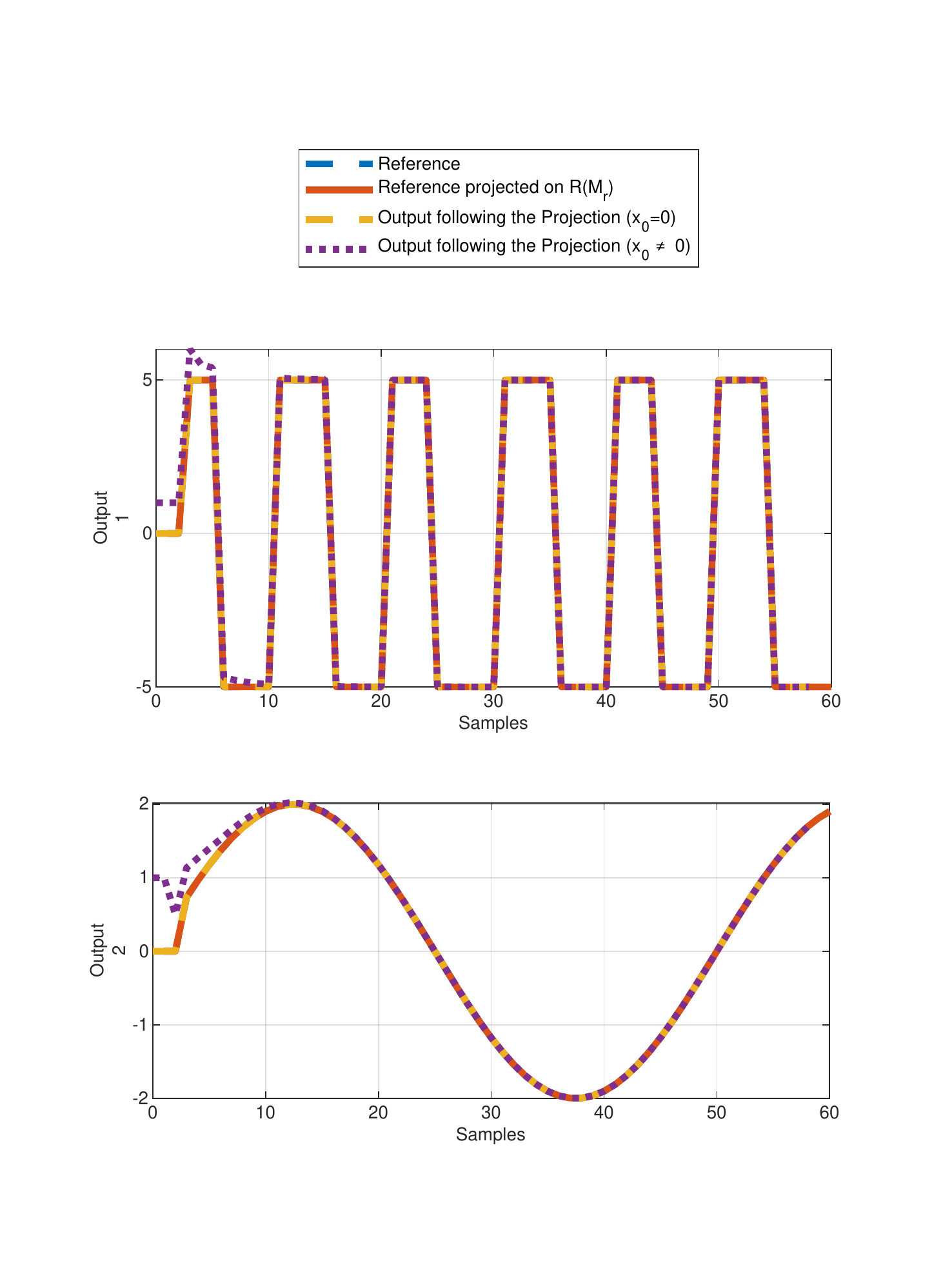}
\caption{Plots showing tracking of references projected on $\mathcal{R}(M_{r})$ when $\mathcal{Y}_{{\rm ref},r,L} \in \mathcal{R}(M_{r})$. }
\label{projscl2}
\end{figure}

\section{Trackability Indices}\label{sec:indices}
\subsection{Reference Command Trackability Index}
For untrackable systems, since not all reference commands can be tracked, it would be helpful to determine \textit{a priori}, whether a reference command sequence could be tracked exactly by the system or not. Furthermore, it will be helpful to know how close is the achievable tracking performance to a given reference command, if it is not in the trackable set of an untrackable system. To achieve this, we define an index
\begin{equation}
\theta({\mathcal{Y}_{{\rm ref},r,L}}) \triangleq 
\frac{\lVert \mathcal{Y}_{{\rm ref},r,L} \Pi_{\mathcal{R}_{(M_{r})}} \rVert _2}{\lVert \mathcal{Y}_{{\rm ref},r,L} \rVert _2},
\end{equation} called the \underline{reference command trackability index}, which indicates the achievable tracking performance of the system. We assume $x_0=0$ for simplicity.  We could have the following cases for  $\theta({\mathcal{Y}_{{\rm ref},r,L}})$:
\begin{itemize}
\item $\theta({\mathcal{Y}_{{\rm ref},r,L}})=1$: When the index $\theta({\mathcal{Y}_{{\rm ref},r,L}})$ for a certain reference command is 1, the reference command lies completely in the trackable set/space of the system, and theoretically, a zero tracking error can be achieved with an appropriate control input.

\item $\theta({\mathcal{Y}_{{\rm ref},r,L}})<1$:  When $\theta({\mathcal{Y}_{{\rm ref},r,L}})<  1$ for a given reference command, $\mathcal{Y}_{{\rm ref},r,L}$ is not in the trackable space/set and exact tracking of all the $l$ components of the reference command cannot be achieved, irrespective of control input used. In other words, minimum achievable tracking error can never be zero.
\item $\theta({\mathcal{Y}_{{\rm ref},r,L}})=0$: When the reference commands lie completely in $\mathcal{N}(M_{r} ^{\top}) = \mathcal{R}(M_{r})^\perp$, which is perpendicular to the trackable space $\mathcal{R}(M_{r})$, the reference command trackability index is zero. The smallest achievable normalized tracking error in this case is larger than the smallest achievable normalized tracking error for all other reference commands in the neighborhood of this reference command, that do not lie exactly in the right null space of $M_{r}$.
\end{itemize}

It may be misconstrued in case of $\theta({\mathcal{Y}_{{\rm ref},r,L}})=0$ that exact tracking of reference commands is not at all possible since the index is zero. However,  exact tracking of some components of the reference commands in the right null space of $M_{r}$ is still possible when using appropriate control inputs, while still having larger normalized error than other reference commands in neighborhood. Consider for example the system \eqref{SS1}, \eqref{SS2} with matrices  \[ 
\begin{split} A=\begin{bmatrix}
         0    &1       & 0       & 0\\
         0         &0    &1         &0\\
         0         &0         &0    &1\\
0.05 &0.1 &0.15 &0.2
\end{bmatrix} ,~B=\begin{bmatrix}
     1     &0\\
     0     &0.8\\
     0     &0.3\\
    0.2     &0.1
\end{bmatrix}\\ \text{~and~}C=\begin{bmatrix}
     1     &0     &0     &0\\
     0    &0     &1     &0\\
     0    &0     &0     &1
\end{bmatrix}. \end{split}\]
The system is not trackable, since $l=3$, and ${\rm rank}(CA^{L-1}B)=2<l$. Hence, $\mathcal{N}(M_{r} ^{\top})$ is not empty. Fig. \ref{P3} shows the outputs of the system for tracking the reference commands in $\mathcal{N}(M_{r} ^{\top})$ making use of control input \eqref{eq:ubar}, and it can be seen that components 2 and 3 of the reference command are tracked exactly while having a large deviation in component 1.
\begin{figure}[!htbp]
\centering
\includegraphics[width=\linewidth]{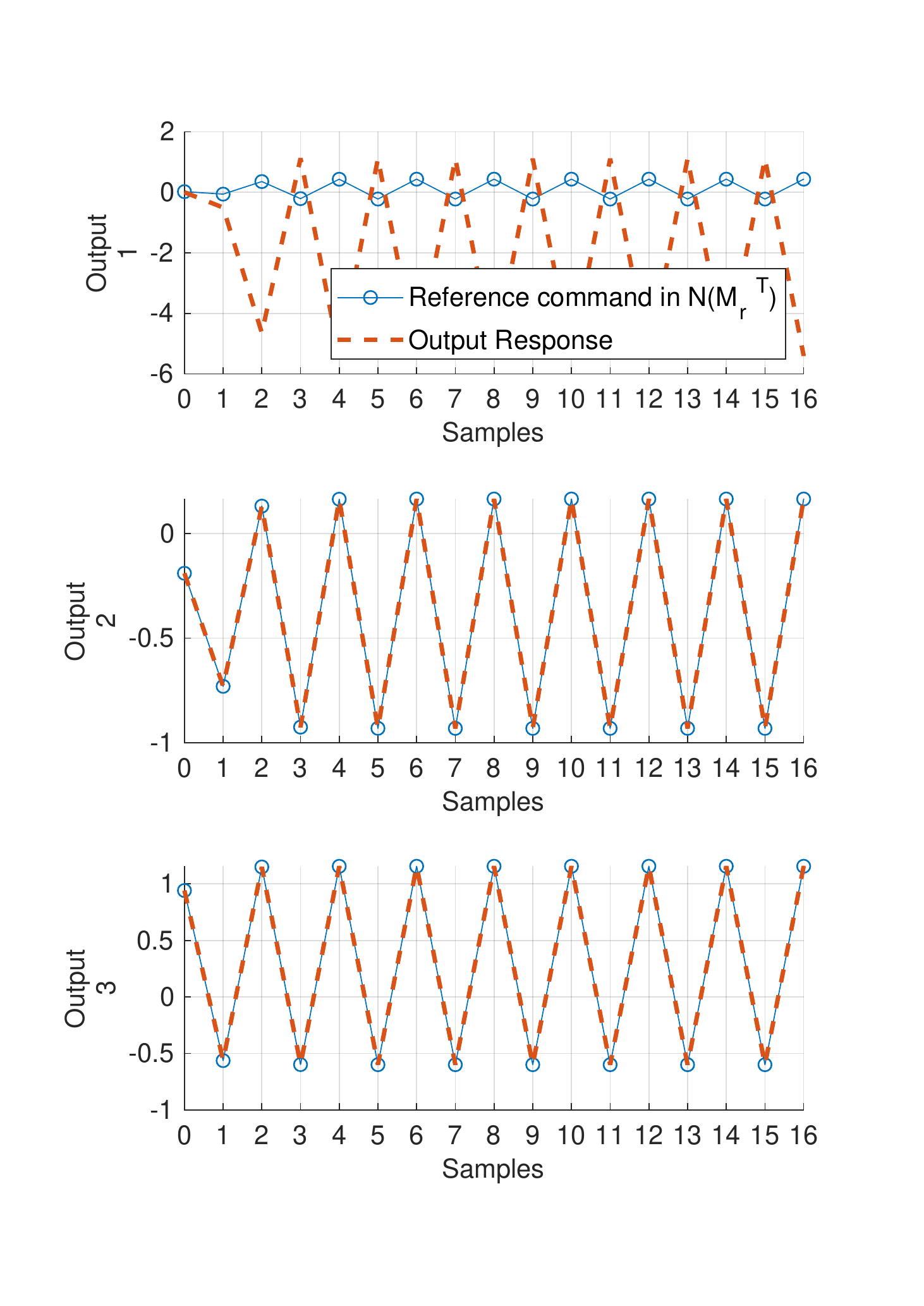}
\caption{Plots showing tracking performances  for $\mathcal{Y}_{{\rm ref},r,L} \in \mathcal{N}(M_{r} ^{\top})$. }
\label{P3}
\end{figure}

We now state the following fact about reference commands in the right null space of $M_{r}$. 
\begin{fact} \label{Yref0}
Let $\mathcal{Y}_{{\rm ref},r,L} \in \mathcal{N}(M_{r} ^{\top})$ and $x_0=0$. Then, for all $\mathcal{U}_{r-L} \in \mathbb{R}^{rm}$,
\[ \lVert \mathcal{E}_{r,L} \rVert_2 \geq \lVert \mathcal{Y}_{{\rm ref},r,L} \rVert_2, \]
where $ \mathcal{E}_{r,L} =  \mathcal{Y}_{{\rm ref},r,L} -  M_{r} \mathcal{U}_{r-L} $.
\end{fact}
\begin{proof}
Since $\mathcal{Y}_{{\rm ref},r,L} \perp M_{r} \mathcal{U}_{r-L}$ , we may write
\begin{equation*}
\begin{split}
\lVert  \mathcal{Y}_{{\rm ref},r,L} -  M_{r} \mathcal{U}_{r-L}  \rVert_2 ^2  &=  \lVert  \mathcal{Y}_{{\rm ref},r,L} \rVert_2 ^2  + \lVert M_{r} \mathcal{U}_{r-L} \rVert_2 ^2 ,\\
\lVert \mathcal{E}_{r,L}  \rVert_2 ^2 &=  \lVert  \mathcal{Y}_{{\rm ref},r,L} \rVert_2 ^2 + \lVert M_{r} \mathcal{U}_{r-L} \rVert_2 ^2,\\
\lVert \mathcal{E}_{r,L}  \rVert_2 ^2 & \geq  \lVert  \mathcal{Y}_{{\rm ref},r,L} \rVert_2 ^2,\\
\text{and}\\
\lVert \mathcal{E}_{r,L}  \rVert_2 & \geq  \lVert  \mathcal{Y}_{{\rm ref},r,L} \rVert_2,\\
\end{split}
\end{equation*}
for all  $\mathcal{U}_{r-L} \in \mathbb{R}^{rm}$.
\end{proof}
$ \mathcal{E}_{r,L} $ represents the tracking error and it is evident from Fact \ref{Yref0} that, irrespective of control input $\mathcal{U}_{r-L}$ used for tracking a reference command $\mathcal{Y}_{{\rm ref},r,L} \in \mathcal{N}(M_{r} ^{\top})$, the lowest achievable tracking error will have a norm  equal to $\lVert \mathcal{Y}_{{\rm ref},r,L} \rVert_2$.
{We now present a fact that gives an upper bound on the lowest achievable tracking error.
\begin{fact}\label{factTI}
For all $\mathcal{Y}_{{\rm ref},r,L} \in \mathbb{R}^{(r-L+1)l}$, $\|\mathcal{E}_{r,L}\|_2 ^2 = (1-\theta(\mathcal{Y}_{\rm ref,r}))\|\mathcal{Y}_{\rm ref,r} \|_2^2.$
\end{fact}
\begin{proof}
Since $\mathcal{Y}_{\rm ref,r}=\mathcal{Y}_{\rm ref,r}\Pi_{\mathcal{R}({M_{r}})} + \mathcal{Y}_{\rm ref,r}\Pi_{\mathcal{N}({M_{r}}^{\top})}$ and  $\mathcal{Y}_{\rm ref,r}\Pi_{\mathcal{N}({M_{r}}^{\top})}^{\top}\mathcal{Y}_{\rm ref,r}\Pi_{\mathcal{R}({M_{r}})} =0,$ we may write
$\mathcal{Y}_{\rm ref,r}^{\top} \mathcal{Y}_{\rm ref,r}\Pi_{\mathcal{R}({M_{r}})} = \|\mathcal{Y}_{\rm ref,r}\Pi_{\mathcal{R}({M_{r}})}\|_2^2. $ Also, since
$\|\mathcal{Y}_{\rm ref,r}-\mathcal{Y}_{\rm ref,r}\Pi_{\mathcal{R}({M_{r}})}\|_2^2 = \|\mathcal{Y}_{\rm ref,r}\|_2^2 - 2 \mathcal{Y}_{\rm ref,r}^{\top} \mathcal{Y}_{\rm ref,r}\Pi_{\mathcal{R}({M_{r}})}+ \|\mathcal{Y}_{\rm ref,r} \Pi_{\mathcal{R}({M_{r}})}\|_2^2$, we have, 
$\|\mathcal{Y}_{\rm ref,r}-\mathcal{Y}_{\rm ref,r}\Pi_{\mathcal{R}({M_{r}})}\|_2^2 = (1-\theta(\mathcal{Y}_{\rm ref,r})\|\mathcal{Y}_{\rm ref,r}\Pi_{\mathcal{R}({M_{r}})}\|_2^2=\|\mathcal{Y}_{\rm ref,r}-{M_{r}}\mathcal{U}_{r-L}\|_2^2 =\|\mathcal{E}_{r,L}\|_2 ^2$.
\end{proof}
}
We present the properties of the reference command trackability index more rigorously in the following.
\begin{fact}\label{factTI}
For all $\mathcal{Y}_{{\rm ref},r,L}\in \mathbb{R}^{(r-L+1)l}$, $0 \leqslant \theta({\mathcal{Y}_{{\rm ref},r,L}}) \leqslant 1$.
\end{fact}

\begin{remark}
The following statements are true with regard to the trackability index $\theta({\mathcal{Y}_{{\rm ref},r,L}})$:
\begin{enumerate}
\item If $\mathcal{Y}_{{\rm ref},r,L} = \mathcal{Y}_{{\rm ref},r,L} \Pi_{\mathcal{R}({M_{r}})}$, then $\mathcal{Y}_{{\rm ref},r,L}\in \mathfrak{T}_r(A,B,C,0)$ and $\theta({\mathcal{Y}_{{\rm ref},r,L}})=1$.
\item For trackable system, $ \theta({\mathcal{Y}_{{\rm ref},r,L}}) =1$ for every $\mathcal{Y}_{{\rm ref},r,L} \in \mathbb{R}^{(r-L+1)l}$.
\item If $\theta({\mathcal{Y}_{{\rm ref},r,L}})<1$ for some $\mathcal{Y}_{{\rm ref},r,L}$, then $\mathcal{Y}_{{\rm ref},r,L}\notin \mathfrak{T}_r(A,B,C,0)$.
\item If $\mathcal{Y}_{{\rm ref},r,L} \Pi_{\mathcal{R}({M_{r}})}=0$, then $\theta({\mathcal{Y}_{{\rm ref},r,L}})=0$. In other words, if $\mathcal{Y}_{{\rm ref},r,L} \in \mathcal{N}(M_{r} ^{\top})$ then, $\theta({\mathcal{Y}_{{\rm ref},r,L}})=0$.
\item $\theta({\mathcal{Y}_{{\rm ref},r,L}})= \theta (\alpha {\mathcal{Y}_{{\rm ref},r,L}})$ for all $\alpha \in \mathbb{R} \backslash \lbrace 0 \rbrace$. That is, $\theta({\mathcal{Y}_{{\rm ref},r,L}})$ is scale invariant.
\end{enumerate}
\end{remark}

It is to be noted that, the trackability index indicates the best possible output behaviour that can be achieved theoretically. This would mean that a certain value of trackability index does not guarantee that the outputs of the system will be as close as determined by the trackability index, since the tracking performance of the system depends on various factors including the controller, modeling errors, disturbances and initial conditions.

\subsection{Component-wise System Trackability Index}
The trackability index $\theta({\mathcal{Y}_{{\rm ref},r,L}})$ defined earlier does not provide an insight about how closely each component of the reference command can be followed, even though it indicates closest possible output behaviour of all the output components combined. Moreover, $\theta({\mathcal{Y}_{{\rm ref},r,L}})$ depends on the reference command sequence $\mathcal{Y}_{{\rm ref},r,L}$. Hence, there is a need to define an index,  which indicates achievable tracking performance for each component of the output vector while also being independent of reference commands.

For a matrix $G \in \mathbb{R}^{l \times m}$, define the unary operation \begin{equation}
{\rm rga}(G) \triangleq G \circ (G ^\dagger)^{\top},
\end{equation} where $\circ$ indicates the Schur-Hadamard or the element-wise product. The matrix ${\rm rga}(G)$ is called the relative gain array (RGA) of the matrix $G$ and was originally defined in \cite{bristol1966rga}. More details about RGA and its properties can be found in \cite{liptak2013process,skogestad2005multivariable}. We now define a vector index {$\vartheta(CA^{L-1}B) \in \mathbb{R}^l$} called \underline{component-wise} \underline{system} \underline{trackability index}, based on RGA of first non-zero Markov parameter {$CA^{L-1}B$}, which indicates how closely each component of output vector can track the respective reference commands. Let $\mathbf{1}_m=\begin{bmatrix}1 &1 &\cdots &1
\end{bmatrix} ^{\top} \in \mathbb{R}^{m}$. Then,
{\begin{equation}
\vartheta(CA^{L-1}B)= \begin{bmatrix}
\vartheta(CA^{L-1}B) ^{(1)}\\
\vartheta(CA^{L-1}B) ^{(2)}\\
\vdots \\
\vartheta(CA^{L-1}B) ^{(l)}
\end{bmatrix} \triangleq {\rm rga}(CA^{L-1}B) \mathbf{1}_m.
\end{equation}}
{$\vartheta(CA^{L-1}B)$} represents component-wise achievable tracking performance for reference commands that do not belong to the trackable space/set.

We now state the following facts about $\vartheta(CA^{L-1}B)$:
\begin{fact}
 $0 \leq \vartheta(CA^{L-1}B) ^{(i)} \leq 1$, for all $i=1,~2,\cdots,~l$.
\end{fact}
\begin{proof}
Given a relative gain array of a matrix, the sum of row or column elements lie in the interval $[0 ~1]$  \cite{skogestad2005multivariable}. Since $\vartheta(CA^{L-1}B) ^{(i)}$ represents the sum of all the elements in the $i^{\rm th}$ row of ${\rm rga}(CA^{L-1}B)$, $0 \leq \vartheta(CA^{L-1}B) ^{(i)} \leq 1$, for all $i=1,~2,\cdots,~l$.
\end{proof}
\begin{remark}
If $ \vartheta(CA^{L-1}B) ^{(i)} =1$, then the $i^{\rm th}$ component of output can exactly track arbitrary reference commands, independently of other components, given that initial conditions are known.
\end{remark}
For reference commands not in the trackable set, a value of $\vartheta(CA^{L-1}B) ^{(i)}$ deviating from 1 indicates the possible deviation of the $i ^{\rm th}$ component of the output from the corresponding component of the reference command. However, reference commands in the trackable subspace can be tracked exactly, irrespective of $\vartheta(CA^{L-1}B)$.

\begin{fact} $\sum_{i=1} ^{l} \vartheta(CA^{L-1}B) ^{(i)} = {\rm rank}(CA^{L-1}B)$.
\end{fact}
\begin{proof}
Noting that $\sum_{i=1} ^{l} \vartheta(CA^{L-1}B) ^{(i)} $ is equal to sum of all the elements in RGA of a matrix and since sum of all the elements in a RGA equals the rank of the matrix \cite{skogestad2005multivariable}, $\sum_{i=1} ^{l} \vartheta(CA^{L-1}B) ^{(i)} = {\rm rank}(CA^{L-1}B)$.
\end{proof}
\begin{fact}
System (\ref{SS1}), (\ref{SS2}) is trackable if and only if $$\vartheta(CA^{L-1}B) ^{(1)}=\vartheta(CA^{L-1}B) ^{(2)}=~\cdots=~ \vartheta(CA^{L-1}B) ^{(l)}=1.$$ Alternatively, $\sum_{i=1} ^{l} \vartheta(CA^{L-1}B) ^{(i)}=l$ if and only if ${\rm rank}(CA^{L-1}B)=l$.
\end{fact}
\subsection{System Trackability Index}
{Based on the vector index of $\vartheta(CA^{L-1}B)$, we may define a scalar index $\Theta$  called the \underline{system trackability index} as
{\begin{equation}
\Theta \triangleq \frac{\vartheta(CA^{L-1}B)^{\top} \mathbf{1}_l}{l} =\frac{{\mathbf{1}_l}^{\top} \operatorname{rga}(CA^{L-1}B) \mathbf{1}_m}{l}
\end{equation}} which depends only on the $CA^{L-1}B$ matrix and its dimensions unlike $\theta({\mathcal{Y}_{{\rm ref},r,L}})$, which depends on all the Markov parameters and the reference command sequence. }
We may now state the following facts for the system trackability index.
\begin{fact}\label{Theta1} $0 \leq \Theta \leq 1 .$ 
\end{fact}
\begin{fact}\label{Theta2}
System (\ref{SS1}), (\ref{SS2}) is trackable if and only if $\Theta =1.$ 
\end{fact}
\begin{fact}\label{Theta3}
$\Theta =1$ if and only if $\vartheta ^{(1)}=\vartheta ^{(2)}=~\cdots=~ \vartheta ^{(l)}=1.$ 
\end{fact}
{Alternatively, $\Theta$ may also be defined as}
{\begin{equation}
\Theta \triangleq \frac{\vartheta(CA^{L-1}B)^{\top} \vartheta(CA^{L-1}B)}{l },
\end{equation}}
{which satisfies the properties stated in Facts \ref{Theta1} through \ref{Theta3} as well.}
\begin{remark}
It must be noted that the system trackability indices $\vartheta(CA^{L-1}B)$ and $\Theta$ give a general indication of the expected tracking behaviour of the system, without a prior knowledge of reference commands or control inputs. The actual tracking behaviour exhibited by the system will however depend on the specific reference command and the controller and may deviate from that indicated by the system trackability indices based on how close or away the reference commands lie from the trackable space/set.
Once the reference command is known, the reference command trackability index $\theta({\mathcal{Y}_{{\rm ref},r,L}})$ provides a more accurate picture of the degree of tracking possible. These three trackability indices when used together give a better insight into the tracking behavior of the system.
\end{remark}
 \section{Other system properties and their relationship with Trackability}\label{sec:venn}
The relationship between trackability and other system properties like controllability, observability, input and state observability \cite{palanthandalam2007unbiased}, and output controllability is summarized with the help of a Venn diagram in Fig. \ref{venn}. Examples corresponding to each region numbered in the Venn diagram are provided in the Appendix. 
In the following, we discuss facts that elaborate the relations between properties depicted in the Venn diagram. {For system \eqref{SS1}, \eqref{SS2} with a delay $L$, let
\begin{equation}
\label{ctrb}
Q_{{\rm sc},r}\triangleq \begin{bmatrix}
B & AB & A^2B & A^3B & \cdots &A^{r-1}B
\end{bmatrix},
\end{equation}
  \begin{equation}
\label{Cctrb}
Q_{{\rm oc},r}\triangleq \begin{bmatrix}
CB & CAB & CA^2B & CA^3B & \cdots &CA^{r-1}B
\end{bmatrix} = C Q_{{\rm sc},r}
\end{equation}
and 
\begin{equation}
\label{obsv}
Q_{{\rm so},r} \triangleq [C^\top~A^\top C^\top~{(A^\top)}^2 C^\top \cdots ~{(A^\top)}^{r-1} C^\top]^\top.
\end{equation}  Then, $Q_{{\rm sc},n}$, $Q_{{\rm so},n}$ and $Q_{{\rm oc},n}$ represent the state controllability, state observability and the output controllability matrices respectively for the system (\ref{SS1}), (\ref{SS2}). We refer to systems that are not state controllable as state uncontrollable.}
{Also define  \begin{equation}
\Psi_r \triangleq \begin{bmatrix}
Q_{{\rm so},L-1} & 0\\
\Gamma_r & M_r
\end{bmatrix}
\end{equation} originally defined in \cite{palanthandalam2007unbiased}, where
 $\Psi_r \in \mathbb{R}^{(r+1)l \times (n+rm)}$.
}
In the following, we establish that {systems with delay $L$} cannot be trackable, input and state observable, and state uncontrollable. In other words, we establish that set number 18 in Fig. \ref{venn} is empty. Furthermore,  we also state a few results about trackable systems and output controllable systems.

\begin{figure*}[!htbp]
\centering
\includegraphics[width=0.99\linewidth]{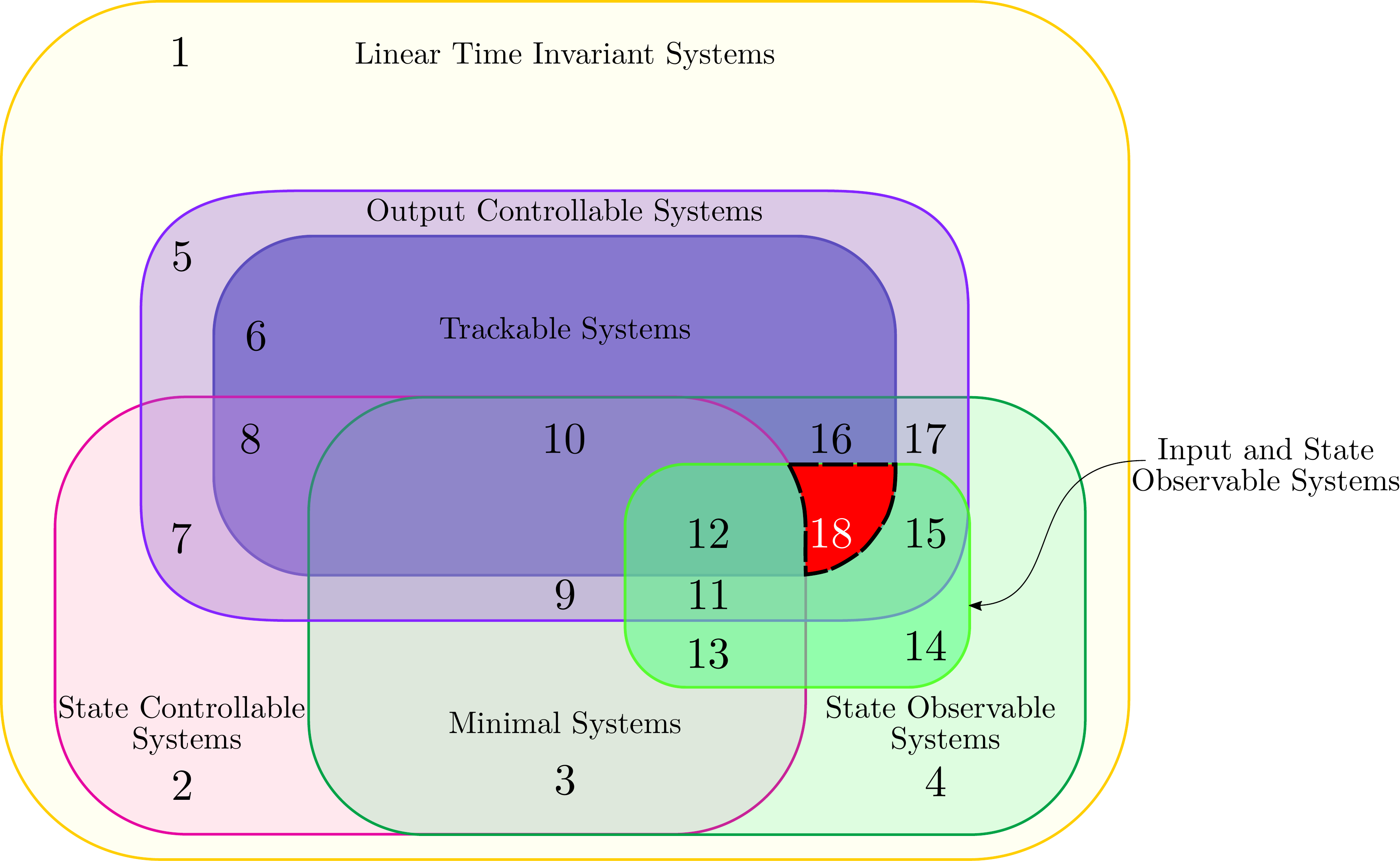}
\caption{Venn diagram showing relations between different systems properties}
\label{venn}
\end{figure*}
\begin{lemma}\label{fact3v}
If the system (\ref{SS1}), (\ref{SS2}) {with a delay $L$} is trackable, then the system is output controllable.
\end{lemma}
\begin{proof} 
Since (\ref{SS1}), (\ref{SS2}) is trackable, ${\rm rank}(CA^{L-1}B)=l$, and $L \leq n$,  ${\rm rank}(Q_{{\rm oc},r})={\rm rank}\left(\begin{bmatrix}CB & CAB &\cdots & CA^{r-1}B \end{bmatrix}\right)=l.$
\end{proof}
The converse of Lemma \ref{fact3v} is not true, since there exist output controllable systems that are not trackable. See counterexamples 5, 7, 9, 11, 15, and 17. We state the contrapositive of Lemma \ref{fact3v} as a corollary.

\begin{corollary}\label{rem_t}
If  system (\ref{SS1}), (\ref{SS2}) is not output controllable, then it is not trackable.
\end{corollary}
{
\begin{lemma}\label{lmn}
Let {system \eqref{SS1}, \eqref{SS2} with delay $L=1$} be input and state observable and trackable. Then, $l=m=n$.
\end{lemma}
\begin{proof}
Since the system with $L=1$ is trackable, ${\rm rank}(M_r)=rl$ and therefore, $\operatorname{rank}\left( \begin{bmatrix}
 0\\
 M_r
\end{bmatrix} \right)=rl$.
 Furthermore, since the system is input and state observable,
${\rm rank}\left( \begin{bmatrix}
 0\\
 M_r
\end{bmatrix} \right)=rm$. Therefore, ${\rm rank}\left( \begin{bmatrix}
 0\\
 M_r
\end{bmatrix} \right)=rm=rl$, that is, $l=m$ and  $l=m=n$ follows from Proposition II.2 in \cite{palanthandalam2009subspace}.
\end{proof}}
\begin{thrm}\label{Thm2}
The set of systems that are trackable, input and state observable, and state uncontrollable is empty.
\end{thrm}
\begin{proof}
Let (\ref{SS1}), (\ref{SS2}) be  input and state observable, and trackable but not state controllable.
Following Lemma \ref{lmn}, since (\ref{SS1}), (\ref{SS2}) is trackable and input and state observable, $l=m=n$. Further, since the (\ref{SS1}), (\ref{SS2}) is not controllable, ${\rm rank}(Q_{{\rm sc},r})<n$ and we have ${\rm rank}(Q_{{\rm oc},r})={\rm rank}(C Q_{{\rm sc},r}) \leq {\rm rank}(Q_{{\rm sc},r})<l=m=n$.
Therefore, the system (\ref{SS1}), (\ref{SS2}) is trackable but not output controllable, thus contradicting Lemma \ref{fact3v} and proving that  intersection of the set of trackable, state uncontrollable and input and state observable systems is empty. 
\end{proof}
In the following remark, alternative statements to Theorem \ref{Thm2} are enlisted.
\begin{remark} {For system \eqref{SS1}, \eqref{SS2}}, the following statements are equivalent to Theorem \ref{Thm2}:
\begin{enumerate}
\item Let \eqref{SS1}, \eqref{SS2} be trackable and input and state observable. Then, it is state controllable.
\item Let \eqref{SS1}, \eqref{SS2} be input and state observable and state uncontrollable. Then, it is not trackable.
\item Let \eqref{SS1}, \eqref{SS2} be trackable and  state uncontrollable. Then, it is not input and state observable.
\end{enumerate}
\end{remark}
In the remainder of this subsection, we present a few results about output controllability of the system.
\begin{remark} \label{minrem}
Minimal systems are not necessarily output controllable.
\end{remark}
Examples 3 and 13 in the Appendix validate Remark \ref{minrem}.
\begin{fact}\label{fact5}
If the system (\ref{SS1}), (\ref{SS2}) is output controllable, then ${\rm rank}(C)=l$.
\end{fact}
\begin{proof} Since the system is output controllable,
${\rm rank}(Q_{{\rm oc},r})={\rm rank}(C Q_{{\rm sc},r})=l$.
Also, ${\rm rank}(C Q_{{\rm sc},r}) \leq {\rm rank}(C)$ and therefore, $l \leq {\rm rank}(C)$.
Since $C \in \mathbb{R}^{l \times n}$ and since $l \leq n$ the maximum possible rank for matrix $C$ is $l$. Hence, ${\rm rank}(C)=l$.
\end{proof}

\begin{fact}
If the system (\ref{SS1}), (\ref{SS2}) is trackable, then ${\rm rank}(C)=l$.
\end{fact}
\begin{proof}
The proof is immediate from Lemma \ref{fact3v} and Fact \ref{fact5}.
\end{proof}

\begin{fact}\label{fact13a}
Let system (\ref{SS1}), (\ref{SS2}) be minimal and ${\rm rank}(C)= l$. Then, the system is output controllable.
\end{fact}
\begin{proof}
 Assume that (\ref{SS1}), (\ref{SS2}) is minimal, has ${\rm rank}(C)= l$ but is not output controllable. Then, ${\rm rank}(C Q_{{\rm sc},r})={\rm rank}(Q_{{\rm oc},r})< l$. 
Since the system is minimal, there exists $Q_{{\rm sc},r} ^\dagger \in \mathbb{R}^{rm \times n}$ such that $Q_{{\rm sc},r} Q_{{\rm sc},r}^\dagger = I_n$. Further, ${\rm rank}(C)={\rm rank}(C Q_{{\rm sc},r} Q_{{\rm sc},r} ^\dagger) \leq {\rm rank}(C Q_{{\rm sc},r}) < l$ and therefore, ${\rm rank}(C)<l$ which contradicts the assumption. Hence, minimal systems with ${\rm rank}(C)= l$ are output controllable.
\end{proof}

\subsection{Number of inputs and outputs}
It must be noted that systems with $l>m$ are not trackable since they do not satisfy conditions in Theorem \ref{TrackThm1} and control inputs for exact tracking of arbitrary reference commands do not exist for these systems. For trackable systems with $l=m$, it is easy to see that control input is unique for tracking a given reference command. 

We now make remarks about trackability as a system property with regard to zeros of the system.
\subsection{Multivariable zeros and trackability}
{Define a matrix  $\tilde{M}_{r} $ such that
\[ \tilde{M}_{r} \triangleq \begin{bmatrix}
CA^{L-1}B & CA^LB & \cdots & CA^{r-1}B\\
0 &CA^{L-1}B  & \cdots & CA^{r-2}B\\
\vdots & &\ddots & \vdots \\
0 & \cdots &0 & CA^{L-1}B
\end{bmatrix} \in \mathbb{R}^{(r-L+1)l \times (r-L+1)m}.\] 
The following fact may be stated therefore.
\begin{fact}\label{Mrtilde}
 ${\rm rank}(\tilde{M}_{r})={\rm rank}({M_{r}})$.
\end{fact}
\begin{proof}
See Appendix \ref{ProofMrTilde}.
\end{proof}
}
We now recall from \cite{palanthandalam2009subspace} the matrix $\Phi_r \in \mathbb{R}^{ (n+rl) \times (r+1)m}$, defined as 
\begin{equation}\label{Phi_r}
\Phi_r \triangleq \left[\begin{array}{c c c c c c}
B & & AB & A^2B &\cdots & A^r B\\ \hhline{------}
0 &\temp & CB & CAB & \cdots & CA^{r-1}B\\
0 &\temp & 0 &CB  & \cdots & CA^{r-2}B\\
\vdots  &\temp & \vdots & &\ddots & \vdots \\
0 &\temp & 0 & \cdots &0 & CB
\end{array}\right]=\left[
\begin{array}{c}
Q_{{\rm sc},r}\\ \hhline{-}
0 ~ |~\tilde{M}_{r} \\ 
\end{array}\right],
\end{equation}
for a system \eqref{SS1}, \eqref{SS2} {with delay $L=1$,} to restate a theorem from \cite{palanthandalam2009subspace} here as a fact, {based on the assumption that the system in minimal.}
\begin{fact} 
\label{mvZerosFact}
The following statements are equivalent:
\begin{enumerate} [label=\roman*)]
\item Either $l <m$ or $l=m=n$, the triplet $(A,B,C)$ in the system (\ref{SS1}), (\ref{SS2}) has no invariant zeros.
\item ${\rm rank} \left( \Phi_{n-1} \right) = n+ (n-1)l$.
\end{enumerate}
\end{fact}
The following result can now be stated.
{
\begin{proposition}\label{zeroFact}
If ${\rm rank} \left( \Phi_{n-1} \right) = n+ (n-1)l$, the system is trackable with a delay $L=1$. 
\end{proposition}
\begin{proof}
Following Facts \ref{mvZerosFact} and \ref{Mrtilde}, if ${\rm rank} \left( \Phi_{n-1} \right) = n+ (n-1)l$, then ${\rm rank} \left( \tilde{M}_{n-1} \right)={\rm rank} \left( {M}_{n-1} \right)=(n-l)l$ and from Lemma {\ref{lemmaA1}}, it follows that ${\rm rank }(CB)=l$.
\end{proof}}
{It thus follows from Fact \ref{mvZerosFact} and Proposition \ref{zeroFact} that a {minimal} system with the triplet $A,~B$ and $C$ and {delay $L=1$}, having either $l<m$ or $l=m=n$, and not having { invariant zeros}, is trackable.} However, the converse is not necessarily true, since there exist {minimal} systems with {$L=1$ delay}, which are  trackable and still have invariant zeros. For example, the system (\ref{SS1}), (\ref{SS2}) having matrices \[ 
\begin{split} A=\begin{bmatrix}
0.2 &0.3 &0.4\\
1 &0.1 &0\\
0 &0.8 &0.6
\end{bmatrix},~B=\begin{bmatrix}
1 & 0\\
0 & 0\\
0 &1
\end{bmatrix}\\ \text{~and~}C=\begin{bmatrix}
1 & 0 &0
\end{bmatrix} \end{split}\]
is minimal, {with $l<m$} and trackable, but has an invariant zero at $0.1$.
 \subsection{Other remarks}
Non-minimum phase systems have unstable zero dynamics and system inversion may not provide bounded control inputs for tracking reference commands.  If non-minimum phase systems are trackable, exact tracking of given reference commands is possible, that is, Theorem \ref{TrackThm1} holds for non-minimum phase systems, even if boundedness of control inputs cannot be guaranteed. However, there is certainly a possibility for defining trackability in the sense of bounded inputs and bounded outputs, which the presented work considers beyond its scope.

Theorem \ref{TrackThm1} is applicable to both, stable or unstable systems,  since stability is a property independent of ${\rm rank}(CA^{L-1}B)$. For an unstable system which is stabilizable as well as trackable, there exists a $K$ such that ${\rm sp.rad.}(A-BK)<1$  and control inputs can be determined for simultaneous stabilization and  tracking of reference commands for the system $$x_{k+1}=(A-BK)x_k + B {u}_k,$$
$$y_k=C x_k.$$  
However, it is important to note that trackability does not imply stabilizability. That is, there exist systems which are trackable but have states that are unstable as well as uncontrollable and cannot be stabilized.

The property of trackability is invariant under similarity transformation as is expected.
Let $\hat{A}=TAT^{-1}, ~\hat{B}=TB \text{ and }\hat{C}=CT^{-1}$ be matrices similar to $A, ~B \text{ and }C$ respectively under similarity transformation represented by invertible matrix $T \in \mathbb{R}^{n \times n}$. Then, $\hat{C} \hat{A}^{L-1} \hat{B} = CA^{L-1}B$ and $\hat{M}_{r}=M_{r}$. 

{While trackability is defined here for systems with a delay of $L$, it is also worth noting that it is possible for a system to have different delays in different input-output channels and hence not have a {\it uniform} delay $L$. The current framework allows for this generalization, however, the details of such a generalization are left for a future work.}


\section{Conclusion}\label{sec:conclusion}
Tracking behavior of system outputs was investigated as an existence problem with the help of basic tools in linear algebra, by setting up the definitions of trackability sets and trackable systems.  Trackable sets give information about reference commands that the dynamical system can follow, which is more useful than knowing if a system is right invertible or not. It was further noted that initial conditions of the system influence the set of reference commands the system can track.  Furthermore, for tracking untrackable reference commands, a workaround in the form of projections was proposed, to achieve tracking performance permitted by capabilities of the system, while keeping the tracking error minimum. Indices to give a sense of expected tracking performance from a system, even before controller design or control synthesis, were presented. The definitions of trackability indices suggest that tracking performance depends not only on the reference commands, but also on the interaction between the input and output variables of the system. A simple rank test to determine if system is trackable or not was proposed, while putting into context the rank tests given in earlier works that discussed similar concepts. It was shown that trackability is a different notion from state controllability and output controllability. More so, it was shown that output controllability is a necessary but not a sufficient condition for a system to track arbitrary reference commands. A Venn diagram summarizing the interconnections between fundamental system properties was presented along with examples and facts for supporting claims about relation of trackability with other properties. Even though the notion of trackability is presented here for discrete-time systems, it is thought that the concept is applicable to continuous-time systems, \textit{mutatis mutandis}. It is also felt that, by defining the property of trackability in an LTI framework, taking a linear algebraic approach, a larger audience would able to gain access to the notion of trackability and tools based on it. 
{Extensions to account for systems having different delays in different input-output channels can be taken up as future work.}

\newpage
\section*{References}
\bibliographystyle{elsarticle-num} 
\bibliography{ref}
\newpage
\appendix
\section*{Appendices}
\renewcommand{\thesubsection}{\Alph{subsection}}
\renewcommand{\theequation}{\thesubsection.\arabic{equation}}

\subsection{Theorem \ref{TrackThm1} and earlier literature}\label{Erl} 
{Theorem \ref{TrackThm1} is closely related to several results in the literature, and is similar to matrix rank tests proposed in earlier papers, most notably in the works \cite{brockett1965reproducibility,sain1969invertibility,kurek2002trackability}. We briefly recount a few of the earlier works that discuss the ability of systems to follow reference commands, often discussed in different frameworks, with different terminologies and assumptions, to give the reader an overview of the literature.}

A  rank condition for checking functional reproducibility of a system - a similar notion to trackability, is given in \cite{brockett1965reproducibility}, in terms of full row rank of a matrix defined as \[{\mathcal{M}_q} \triangleq \begin{bmatrix}
CB & CAB & \cdots & CA^{q-1}B & CA^{q}B & \cdots &CA^{2q-1}B\\
0 &CB  & \cdots & CA^{q-2}B &CA^{q-1}B & \cdots &CA^{2q-2}B\\
\vdots &\vdots &\ddots &\vdots &\vdots &\cdots &\vdots\\
0 & \cdots &0 & CB & CAB & \cdots & CA^{q-1}B
\end{bmatrix}. \]
It is obvious that, if ${{\rm rank}(CB)}=l$, then ${\rm rank}({\mathcal{M}_q})={q}l$. Furthermore, it is easy to see that the converse of this statement is not true \cite{albrecht1981reproducibility}. It is worth noting that Brockett and Mesarovi\`c \cite{brockett1965reproducibility}, first introduced the notion of reproducibility - the ability of system to achieve with its outputs the desired reference commands. Brockett and Mesarovi\`c also discuss the notions of (uniform) functional reproducibility, asymptotic reproducibility and point-wise reproducibility in the context of tracking desired trajectories.

 In \cite{sain1969invertibility}, authors have discussed invertibility of systems from the perspective of both existence and construction of inverses, while giving rank conditions for existence of $L$ - delay inverses. Furthermore, the authors have related their results to the notion of functional reproducibility discussed earlier in \cite{brockett1965reproducibility} by defining 
 a matrix $\mathcal{M}_{D_i}$ defined as\[ {\mathcal{M}}_{D_i} \triangleq \begin{bmatrix}
D& CB & CAB & \cdots & CA^{i-1}B\\
0 &D &CB  & \cdots & CA^{i-2}B\\
\vdots &\vdots &\ddots &\vdots &\vdots\\
0 &0 & \cdots &D & CB\\
0 &0 & \cdots &0 & D
\end{bmatrix} \] and specifying the rank test ${\rm rank}({\mathcal{M}}_{D_i})-{\rm rank}({\mathcal{M}}_{D_{i-1}})=l$ as a condition for checking functional reproducibility (trackability) of the system.
More recently, in \cite{kurek2002trackability} a matrix \[
\Delta(0:j) \triangleq \begin{bmatrix}
D &0  &0 &\cdots &0 \\
CB &D &0 &\cdots &0 \\
CAB &CB &D &\cdots &0 \\
\vdots &\vdots &\vdots  &\ddots &0 \\
CA^{j-1}B &CA^{j-2}B &\cdots &CB &D\\
\end{bmatrix}.\] is defined for determining the trackability of the system, by checking either of the following equivalent conditions.
\begin{enumerate}[label=\roman*)]
\item ${\rm rank}(\Delta(0:n))-{\rm rank}(\Delta(0:n-1))=l$
\item ${\rm rank}(\Delta(0:n))$ has full row rank
\item ${\rm rank}(\Delta(n:2n+1))$ has full row rank where,

\end{enumerate}
It may be noted that the matrices ${\mathcal{M}}_{D_i}$  \cite{sain1969invertibility} and $\Delta(0:j)$ \cite{kurek2002trackability} are closely related to each other, and are constructed based on Markov parameters and the block triangular structure. Further, a notion of bounded-output-bounded-input (BOBI) trackability is discussed in \cite{kurek2002trackability} as well.  
Theorem \ref{TrackThm1} discussed earlier provides a simple rank test based on one Markov parameter $CA^{L-1}B$ for systems with delay $L$. In fact, the present paper takes a different approach from the works  \cite{brockett1965reproducibility,sain1969invertibility,kurek2002trackability} by defining the property of trackability through the notion of trackable and untrackable sets which provide a different insight into the tracking performance of the systems.

In \cite{germani1983functional}, a notion of $\epsilon$-controllability was presented, which the authors state, is a generalization of functional output controllability.
The authors give a rank condition on matrix $CBB^{\top} C^{\top}$ which is equivalent to  ${\rm rank}(CA^{L-1}B)=l$ with $L=1$, for determining functional output controllability.
The paper by Aoki \cite{aoki}, has discussed output controllability and perfect output controllability with respect to rank condition on transfer function matrix and the Tinbergen condition $(m \geq l)$. 
Furthermore, the author extends the definitions and condition for perfect output controllability to nonlinear system by stating that the nonlinear system is perfectly output controllable if its linearization is perfectly output controllable.
Wohltmann, in \cite{wohltmann1985target} has presented the definitions for target path controllability and rank conditions to determine target path controllability. 
It is worthwhile to note that the rank conditions presented in \cite{wohltmann1985target} are same as for ${\rm rank}(CA^{L-1}B)$ given by Theorem \ref{TrackThm1} for systems with delay $L=1$. 
The property of target path controllability for linear time varying case based on rank conditions on individual system matrices and fundamental subspaces associated with them is discussed in \cite{engwerda1988control}.
Furthermore, the author also defines a notion of admissible target functions, which is similar to the definitions of trackable sets. 
Notions of admissible states and inputs are also discussed in a similar manner. In \cite{marro2002convolution}, geometric conditions for right invertibility of discrete time LTI systems are discussed.
In reference \cite{kovalev1998criteria},  definitions for functional output controllability and invertibility of nonlinear systems along with criteria of `zero output defect' to determine functional output controllability are discussed. Also, the Tinbergen condition - $m \geq l$ \cite{tinbergen1952theory} is stated as necessary condition for functional output controllability.
\newpage
\subsection{{Venn Diagram Examples:}}
We list in the following, 17 examples for systems corresponding to the sets indicated by numbers in Venn diagram shown in Fig. \ref{venn}. It may be recalled that, set numbered 18 is a null set and therefore no examples exist for this set. Note that, in the examples that follow, $r=n$ is chosen.
\subsubsection*{Example 1}
$ A =
\begin{bmatrix}
     0     &1     &0\\
     0     &0     &1\\
     0     &0     &0
\end{bmatrix},~ B =
\begin{bmatrix}
     1     &0\\
     0     &1\\
     0     &0
\end{bmatrix} \text{ and }
C =
\begin{bmatrix}
0     &0     &1
\end{bmatrix}
$  
 
$
 Q_{{\rm sc},n} =
\begin{bmatrix}

     1     &0     &0     &1     &0     &0\\
     0     &1     &0     &0     &0     &0\\
     0     &0     &0     &0     &0     &0
\end{bmatrix}
$, 
$
Q_{{\rm so},n} =
\begin{bmatrix}
     0     &0     &1\\
     0     &0     &0\\
     0     &0     &0
\end{bmatrix}
$, 
$
 Q_{{\rm oc},n} =
\begin{bmatrix}
     0     &0     &0     &0     &0     &0
\end{bmatrix}
$.  Note that
$
CB = CAB = CA^2B=
\begin{bmatrix}

     0     &0
\end{bmatrix}
$ and therefore, {$L$ is not defined for the system.}
Also, ${\rm rank}(\Psi_r)=1$ and  $n+rm= 9$.

The system is
   \underline{not} controllable,
   \underline{not} state observable,
   \underline{not} minimal,
   \underline{not} output controllable,
   \underline{not} input and state observable,
   \underline{not} trackable.

\subsubsection*{Example 2}
  $ A =
\begin{bmatrix}
     0     &1\\
     0     &1
\end{bmatrix}
\text{, }
B =
\begin{bmatrix}
0\\     1
\end{bmatrix}
\text{ and }     
C =
\begin{bmatrix}
     0    &10\\
     0     &1
\end{bmatrix}
$
  Then,
$
 Q_{{\rm sc},n} =
\begin{bmatrix}
     0     &1\\
     1     &1
\end{bmatrix}
$, 
$
Q_{{\rm so},n} =
\begin{bmatrix}
     0    &10\\
     0    & 1\\
     0    &10\\
     0     &1
\end{bmatrix}
$, 
$
Q_{{\rm oc},n} =
\begin{bmatrix}
    10    &10\\
     1     &1
\end{bmatrix}
$ and
$
CB =
\begin{bmatrix}
    10\\
     1
\end{bmatrix}
$, which implies $L=1$.
Also, 
${\rm rank}(\Psi_r) =   3$ and $n+rm = 4$.
 
The system is  controllable, but
   \underline{not} state observable,
   \underline{not} minimal,
   \underline{not} output controllable,
   \underline{not} input and state observable,
   \underline{not} trackable.

\subsubsection*{Example 3}
 
$
A =
\begin{bmatrix}
     0     &0\\
     1     &0
\end{bmatrix}
$, 
$
B =
\begin{bmatrix}
     1\\
     0
\end{bmatrix}
$ and
$
C =
\begin{bmatrix}
     2    &-4\\
    -1     &2
\end{bmatrix}
$
  Then,
$
 Q_{{\rm sc},n} =
\begin{bmatrix}
     1     &0\\
     0     &1
\end{bmatrix}
$, 
$
Q_{{\rm so},n} =
\begin{bmatrix}
     2    &-4\\
    -1     &2\\
    -4     &0\\
     2     &0
\end{bmatrix}
$, 
$
Q_{{\rm oc},n} =
\begin{bmatrix}
     2    &-4\\
    -1     &2
\end{bmatrix}
$ and
$
CB =
\begin{bmatrix}
     2\\
    -1
\end{bmatrix}
$, which implies, $L=1$.
Also, ${\rm rank}(\Psi_r) =   3$ and  $n+rm = 4$.
 
The system is
  controllable,
  state observable,
  minimal, but   \underline{not} output controllable,
   \underline{not} input and state observable,
   \underline{not} trackable.
 
\subsubsection*{Example 4}
 $
A =
\begin{bmatrix}
     0.1     &0\\
     0     &0
\end{bmatrix}
$, 
$
B =
\begin{bmatrix}
     1     &0\\
     0     &0
\end{bmatrix}
$ and
$
C =
\begin{bmatrix}
     1     &0\\
     0     &1
\end{bmatrix}
$
 
Then,
$
 Q_{{\rm sc},n} =
\begin{bmatrix}
     1     &0     &0.1     &0\\
     0     &0     &0     &0
\end{bmatrix}
$, 
$
Q_{{\rm so},n} =
\begin{bmatrix}
     1     &0\\
     0     &1
     0.1     &0
     0     &0
\end{bmatrix}
$, 
$
 Q_{{\rm oc},n} =
\begin{bmatrix}
     1     &0     &0.1     &0\\
     0     &0     &0     &0
\end{bmatrix}
$ and
$
CB =
\begin{bmatrix}
     1     &0\\
     0     &0
\end{bmatrix}
$, which implies, $L=1$. 
Also, ${\rm rank}(\Psi_r)=4$ and $n+rm = 6$.
 
The system is state observable, but
   \underline{not} controllable,
   \underline{not} minimal,
   \underline{not} output controllable,
   \underline{not} input and state observable,
   \underline{not} trackable.

\subsubsection*{Example 5}
 
$
A =
\begin{bmatrix}
     1     &0     &0 &1\\
     0     &0     &1 &0\\
     0     &1     &0 &0\\
     0     &0     &0 &0
\end{bmatrix}
$, 
$
B =
\begin{bmatrix}
     0     &0\\
     1     &1\\
     0     &0\\
     0     &0
\end{bmatrix}
$ and
$
C =
\begin{bmatrix}
     0     &1     &0  &0\\
     0     &0     &1 &0\\
\end{bmatrix}
$
 
Then,
$
 Q_{{\rm sc},n} =
\begin{bmatrix}
     0     &0     &0     &0     &0     &0     &0     &0\\
     1     &1     &0     &0     &1     &1     &0     &0\\
     0     &0     &1     &1     &0     &0     &1     &1\\
     0     &0     &0     &0     &0     &0     &0     &0
\end{bmatrix}
$, 
$
Q_{{\rm so},n} =
\begin{bmatrix}
     0     &1     &0     &0\\
     0     &0     &1     &0\\
     0     &0     &1     &0\\
     0     &1     &0     &0\\
     0     &1     &0     &0\\
     0     &0     &1     &0\\
     0     &0     &1     &0\\
     0     &1     &0     &0
\end{bmatrix}
$, 
$
 Q_{{\rm oc},n} =
\begin{bmatrix}
1     &1     &0     &0     &1     &1     &0     &0\\
     0     &0     &1     &1     &0     &0     &1     &1
\end{bmatrix}
$ and
$
CB =
\begin{bmatrix}
	 1		&1\\
     0      &0
\end{bmatrix}
$, which implies, $L=1$. Also, ${\rm rank}(\Psi_r) = 6$ and $n+rm =  12$.
 
The system is output controllable, but
   \underline{not} controllable,
   \underline{not} state observable,
   \underline{not} minimal,
   \underline{not} input and state observable,
   \underline{not} trackable.

\subsubsection*{Example 6}
$
A =
\begin{bmatrix}
    0.1         &0         &0         &0\\
         0         &0    &0.1         &0\\
         0    &1         &0         &0\\
         0         &0         &0         &0
\end{bmatrix}
$, 
$
B =
\begin{bmatrix}
     0     &0\\
     1     &1\\
     0     &1\\
     0     &0
\end{bmatrix}
$ and
$
C =
\begin{bmatrix}
     0     &0     &1 &0
\end{bmatrix}
$
 
Then,
$
 Q_{{\rm sc},n} =
\begin{bmatrix}
         0         &0         &0         &0         &0         &0         &0         &0\\
    1    &1         &0         &0    &0.1    &0.1         &0         &0\\
         0         &0    &1    &1         &0         &0    &0.1    &0.1\\
         0         &0         &0         &0         &0         &0         &0         &0
\end{bmatrix}
$, 
$
Q_{{\rm so},n} =
\begin{bmatrix}
0         &0    &1         &0\\
         0    &1         &0         &0\\
         0         &0    &0.1         &0\\
         0    &0.1         &0         &0
\end{bmatrix}
$, 
$
 Q_{{\rm oc},n} =
\begin{bmatrix}
0         &0    &1    &1    &0         &0    &0.1    &0.1
\end{bmatrix}
$ and
$
CAB =
\begin{bmatrix}
     1     &1
\end{bmatrix}
$, which implies, $L=2$, since $CB=\begin{bmatrix}
     1     &1
\end{bmatrix}$.
Also ${\rm rank}(\Psi_r)= 5$ and $n+rm  =  12$.

The system is  output controllable, trackable, but
   \underline{not} controllable,
   \underline{not} state observable,
   \underline{not} minimal,
   \underline{not} input and state observable.

\subsubsection*{Example 7}
 $
A =
\begin{bmatrix}
     1     &0     &0     &1\\
     0     &0     &1     &0\\
     0     &1     &0     &0\\
     0     &0     &0     &0
\end{bmatrix}
$, 
$
B =
\begin{bmatrix}
     0     &0\\
     1     &1\\
     0     &0\\
     0     &-1
\end{bmatrix}
$ and 
$
C =
\begin{bmatrix}
    0     &1     &0     &0\\
    0     &0     &1     &0
\end{bmatrix}
$
 
Then,
$$
 Q_{{\rm sc},n} =
\begin{bmatrix}
     0     &0     &0    &-1     &0    &-1     &0    &-1\\
     1     &1     &0     &0     &1     &1     &0     &0\\
     0     &0     &1     &1     &0     &0     &1     &1\\
     0    &-1     &0     &0     &0     &0     &0     &0
\end{bmatrix}
,$$ 
$
Q_{{\rm so},n} =
\begin{bmatrix}
     1     &0     &0     &0\\
     0     &0     &0     &1\\
     1     &0     &0     &1\\
     0     &0     &0     &0\\
     1     &0     &0     &1\\
     0     &0     &0     &0\\
     1     &0     &0     &1\\
     0     &0     &0     &0
\end{bmatrix}
, $
$
 Q_{{\rm oc},n} =
\begin{bmatrix}
     0     &0     &0    &-1     &0    &-1     &0    &-1\\
     0    &-1     &0     &0     &0     &0     &0     &0
\end{bmatrix}
$ and
$
CB =
\begin{bmatrix}
     0     &0\\
     0     &-1
\end{bmatrix}
$, which implies, $L=1$.
Also, ${\rm rank}(\Psi_r)= 6$ and $n+rm = 12$.
 
The system is controllable, output controllable,
   \underline{not} state observable,
   \underline{not} minimal,
   \underline{not} input and state observable,
   \underline{not} trackable.
 
\subsubsection*{Example 8}
  $
A =
\begin{bmatrix}
    0.1         &0         &0\\
         0         &0    &0.1\\
         0    &1         &0
\end{bmatrix}
$, 
$
B =
\begin{bmatrix}
     0     &1\\
     1     &1\\
     0     &1
\end{bmatrix}
$ and
$
C =
\begin{bmatrix}
     0     &0     &1
\end{bmatrix}
$
 
Then,
$
 Q_{{\rm sc},n} =
\begin{bmatrix}
         0    &1         &0    &0.1         &0    &0.01\\
    1        &0         &0    &0.1    &0.1         &0\\
         0    &1    &1         &0         &0    &0.1
\end{bmatrix}
$, 
$
Q_{{\rm so},n} =
\begin{bmatrix}
         0         &0    &1\\
         0    &1         &0 \\
         0         &0    &0.1
\end{bmatrix}
$, 
$
 Q_{{\rm oc},n} =
\begin{bmatrix}
         0    &1    &1    &1        & 0    &0.1
\end{bmatrix}
$ and 
$
CB =
\begin{bmatrix}
     0     &1
\end{bmatrix}
$, which implies, $L=1$. Also, ${\rm rank}(\Psi_r)=  4$ and $n+rm =   9$.

The system is  controllable,   output controllable,  trackable, but
   \underline{not} state observable,
   \underline{not} minimal,
   \underline{not} input and state observable.
 
\subsubsection*{Example 9}
 $
A =
\begin{bmatrix}
     0     &1\\
     0     &0
\end{bmatrix}
$, 
$
B =
\begin{bmatrix}
     1\\
     1
\end{bmatrix}
$ and
$
C =
\begin{bmatrix}
     1    &-1
\end{bmatrix}
$  
Then,\\
$
 Q_{{\rm sc},n} =
\begin{bmatrix}
     0     &0     &0     &0     &0     &0     &1     &2     &3\\
     0     &0     &0     &1     &2     &3     &0     &0     &0\\
     1     &2     &3     &0     &0     &0     &0     &0     &0
\end{bmatrix}
$, 
$
Q_{{\rm so},n} =
\begin{bmatrix}
     1     &0     &0\\
     0     &1     &0\\
     0     &1     &0\\
     0     &0     &1\\
     0     &0     &1\\
     0     &0     &0
\end{bmatrix}
$,\\
$
 Q_{{\rm oc},n} =
\begin{bmatrix}
     0     &0     &0     &0     &0     &0     &1     &2     &3\\
     0     &0     &0     &1     &2     &3     &0     &0     &0
\end{bmatrix}
$ and
$
CAB = \begin{bmatrix}
     0     &0     &0\\
     1     &2     &3
\end{bmatrix}
$, which implies, $L=2$, since $CB= \begin{bmatrix}
     0     &0     &0\\
     0     &0     &0
\end{bmatrix}$. Also, ${\rm rank}(\Psi_r) = 5$ and $n+rm =  12.$
  
The system is controllable, state observable, minimal, output controllable, but
   \underline{not} input and state observable,
   \underline{not} trackable.
 
\subsubsection*{Example 10}
 
$
A =
\begin{bmatrix}
     0     &1\\
     0    & 0
\end{bmatrix}
$, 
$
B =
\begin{bmatrix}
     -1\\
     1
\end{bmatrix}
$ and 
$
C =
\begin{bmatrix}
     1    & 1
\end{bmatrix}
$  
Then,
$
 Q_{{\rm sc},n} =
\begin{bmatrix}
     -1    & 1\\
     1    & 0
\end{bmatrix}
$, 
$
Q_{{\rm so},n} =
\begin{bmatrix}
     1   &  1\\
     0   &  1
\end{bmatrix}
$, 
$
 Q_{{\rm oc},n} =
\begin{bmatrix}
     0     &1
\end{bmatrix}
$ and
$
CAB = 1$ which implies, $L=2$, since $CB=0$. Also, ${\rm rank}(\Psi_r) = 3$ and $n+rm = 4$.

The system is controllable, state observable, minimal,  output controllable, trackable, but
   \underline{not} input and state observable.

\subsubsection*{Example 11}
  $
A =
\begin{bmatrix}
     1     &0     &0\\
     0     &0    & 1\\
     0     &1    & 0
\end{bmatrix}
$, 
$
B =
\begin{bmatrix}
     1     &0\\
     0     &0\\
     0     &1
\end{bmatrix}
$ and 
$
C =
\begin{bmatrix}
     1     &0    & 0\\
     0     &1    & 0\\
     0     &0    & 1
\end{bmatrix}
$
 
Then,
$
Q_{{\rm sc},n} =
\begin{bmatrix}
     1     &0     &1     &0     &1     &0\\
     0     &0     &0     &1     &0     &0\\
     0     &1    & 0     &0     &0     &1
\end{bmatrix}
$, 
$
Q_{{\rm so},n} =
\begin{bmatrix}
     1     &0     &0\\
     0     &1     &0\\
     0     &0     &1\\
     1     &0    & 0\\
     0    & 0     &1\\
     0     &1     &0\\
     1     &0     &0\\
     0     &1     &0\\
     0     &0     &1
\end{bmatrix}
$, \\
$
 Q_{{\rm oc},n} =
\begin{bmatrix}
     1     &0     &1    & 0    & 1    & 0\\
     0    & 0    & 0     &1    & 0    & 0\\
     0     &1     &0     &0    & 0    & 1
\end{bmatrix}
$ and
$
CB =
\begin{bmatrix}
     1     &0\\
     0     &0\\
     0    & 1
\end{bmatrix}
$ which implies, $L=1$. Also, ${\rm rank}(\Psi_r) =  9$ and $n+rm = 9$.

 The system is controllable, state observable, minimal, output controllable,
  input and state observable, but
   \underline{not} trackable.

\subsubsection*{Example 12}
 $
A =
\begin{bmatrix}
     0     &1\\
     0    & 0
\end{bmatrix}
$, 
$
B =
\begin{bmatrix}
     1     &0\\
     0     &1
\end{bmatrix}
$ and 
$
C =
\begin{bmatrix}
     1    & 0\\
     0    & 1
     \end{bmatrix}
$  
Then,
$
Q_{{\rm sc},n} =
\begin{bmatrix}
     1    & 0    &0     &1\\
     0     &1    & 0    & 0
\end{bmatrix}
$, 
$
Q_{{\rm so},n} =
\begin{bmatrix}
     1    & 0\\
     0    & 1\\
     0    & 1\\
     0    & 0
\end{bmatrix}
$, 
$
 Q_{{\rm oc},n} =
\begin{bmatrix}
     1    & 0     &0     &1\\
     0    & 1     &0     &0
\end{bmatrix}
$ and 
$
CB =
\begin{bmatrix}
     1     &0\\
     0     &1
\end{bmatrix}
$ which implies, $L=1$. Also, ${\rm rank}(\Psi_r) = 6$ and $n+rm =  6$.

The system is controllable, state observable, minimal, output controllable,
  input and state observable, trackable.
 
\subsubsection*{Example 13}
  $
A =
\begin{bmatrix}
     0    & 1    & 0\\
     0    & 0    & 1\\
     0    & 0   &  0
\end{bmatrix}
$, 
$
B =
\begin{bmatrix}
     0\\
     0\\
     1
\end{bmatrix}
$ and 
$
C =
\begin{bmatrix}
     1     &0     &0\\
     0     &0    & 1\\
     1    & 0    & 1
\end{bmatrix}
$
 
Then, $
Q_{{\rm sc},n} =
\begin{bmatrix}
     0    & 0    & 1\\
     0    & 1    & 0\\
     1    & 0    & 0
\end{bmatrix}
$, 
$
Q_{{\rm so},n} =
\begin{bmatrix}
     1     &0     &0\\
     0    & 0     &1\\
     1    & 0     &1\\
     0    & 1    & 0\\
     0    & 0     &0\\
     0    & 1     &0\\
     0    & 0    & 1\\
     0    & 0     &0\\
     0    & 0    & 1
\end{bmatrix}
$, 
$
 Q_{{\rm oc},n} =
\begin{bmatrix}
     0     &0     &1\\
     1     &0     &0\\
     1     &0     &1
\end{bmatrix}
$ and 
$
CB =
\begin{bmatrix}
     0\\
     1\\
     1
\end{bmatrix}
$, which implies $L=1$. Also, ${\rm rank}(\Psi_r) = 6$ and $n+rm = 6$.

The system is controllable, state observable, minimal,  input and state observable, but
\underline{not} output controllable,
\underline{not} trackable.

\subsubsection*{Example 14}
 $
A =
\begin{bmatrix}
     1     &0     &0\\
     0     &0     &1\\
     0     &0     &1
\end{bmatrix}
$, 
$
B =
\begin{bmatrix}
     0    & 1\\
     1     &1\\
     1     &1
\end{bmatrix}
$ and 
$
C =
\begin{bmatrix}
     1    & 0     &0\\
     0    & 1    & 0\\
     0    & 0    & 1
\end{bmatrix}
$
 
Then,
$
Q_{{\rm sc},n} =
\begin{bmatrix}
     0     &1    & 0     &1    & 0    & 1\\
     1     &1     &1    & 1    & 1    & 1\\
     1    & 1    & 1    & 1    & 1    & 1
\end{bmatrix}
$, 
$
Q_{{\rm so},n} =
\begin{bmatrix}
     1    & 0     &0\\
     0     &1    & 0\\
     0     &0    & 1\\
     1    & 0    & 0\\
     0    & 0    & 1\\
     0    & 0    & 1\\
     1     &0    & 0\\
     0    & 0    & 1\\
     0    & 0    & 1
\end{bmatrix}
$, 
$
 Q_{{\rm oc},n} =
\begin{bmatrix}
     0     &1     &0    & 1     &0     &1\\
     1     &1     &1    & 1    & 1    & 1\\
     1     &1     &1    & 1    & 1    & 1
\end{bmatrix}
$ and
$
CB =
\begin{bmatrix}
     0     &1\\
     1     &1\\
     1     &1
\end{bmatrix}
$, which implies $L=1$. Also, ${\rm rank}(\Psi_r) =  9$ and $n+rm =  9$.
 
The system is state observable,   input and state observable, but
   \underline{not} controllable,
   \underline{not} minimal,
   \underline{not} output controllable,
   \underline{not} trackable.
 
\subsubsection*{Example 15}
 $
A =
\begin{bmatrix}
     1    &-1     &0\\
     0     &1     &0\\
     0     &0     &1
\end{bmatrix}
$, 
$
B =
\begin{bmatrix}
     1\\
     1\\
     1
\end{bmatrix}
$ and 
$
C =
\begin{bmatrix}
     1    & 0     &0\\
     0    & 0     &1
\end{bmatrix}
$  
Then,
$
 Q_{{\rm sc},n} =
\begin{bmatrix}
     1     &0    &-1\\
     1     &1     &1\\
     1     &1     &1
\end{bmatrix}
$, 
$
Q_{{\rm so},n} =
\begin{bmatrix}
     1    & 0    & 0\\
     0    & 0    & 1\\
     1    &-1     &0\\
     0    & 0     &1\\
     1    &-2     &0\\
     0    & 0     &1
\end{bmatrix}
$, 
$
 Q_{{\rm oc},n} =
\begin{bmatrix}
     1    & 0    &-1\\
     1     &1     &1
\end{bmatrix}
$ and 
$
CB =
\begin{bmatrix}
     1\\
     1
\end{bmatrix}
$, which implies $L=1$. Also, ${\rm rank}(\Psi_r) = 6$ and $n+rm = 6$.

The system is   state observable, output controllable,   input and state observable, but
   \underline{not} controllable,
   \underline{not} minimal,
   \underline{not} trackable.
 
\subsubsection*{Example 16}
  $
A =
\begin{bmatrix}
     0     &1\\
     0     &1
\end{bmatrix}
$, 
$
B =
\begin{bmatrix}
     1\\
     0
\end{bmatrix}
$ and 
$
C =
\begin{bmatrix}
     1     &1
\end{bmatrix}
$  
Then,
$
Q_{{\rm sc},n} =
\begin{bmatrix}
     1     &0\\
     0    & 0
\end{bmatrix}
$, 
$
Q_{{\rm so},n} =
\begin{bmatrix}
     1    & 1\\
     0    & 2
\end{bmatrix}
$, 
$
 Q_{{\rm oc},n} =
\begin{bmatrix}
     1     &0
\end{bmatrix}
$ and 
$
CB = 1$, which implies $L=1$.
Also, ${\rm rank}(\Psi_r)  =  3$ and $n+rm =  4$.

The system is state observable,   output controllable,   trackable, but
   \underline{not} controllable,
   \underline{not} minimal,
   \underline{not} input and state observable.

\subsubsection*{Example 17}
  $
A =
\begin{bmatrix}
     0     &0     &0\\
     1     &0     &0\\
     0     &1    & 0
\end{bmatrix}
$, 
$
B =
\begin{bmatrix}
     0 &0\\
     -1 &0\\
     1 &-1
\end{bmatrix}
$ and 
$
C =
\begin{bmatrix}
     0     &1     &-1\\
     0     &0     &1
\end{bmatrix}
$  
Then,
$
 Q_{{\rm sc},n} =
\begin{bmatrix}
     0     &0     &0     &0     &0     &0\\
    -1     &1     &0     &0     &0     &0\\
     1    &-1    &-1     &1     &0     &0
\end{bmatrix}
$, 
$
Q_{{\rm so},n} =
\begin{bmatrix}
     0     &1    &-1\\
     0     &0     &1\\
     1    &-1     &0\\
     0     &1     &0\\
    -1     &0     &0\\
     1     &0     &0
\end{bmatrix}
$,\\
$
Q_{{\rm oc},n} =
\begin{bmatrix}
    -2     &2     &1    &-1     &0     &0\\
     1    &-1    &-1     &1     &0     &0
\end{bmatrix}
$ and 
$
CB =
\begin{bmatrix}
     -2     &2\\  1   &-1
\end{bmatrix}
$, which implies $L=1$. Also, ${\rm rank}(\Psi_r) = 6$ and $n+rm = 9$.
The system is state observable, output controllable,   trackable but \underline{not} controllable, 
\underline{not} minimal, \underline{not} input and state observable.

\newpage
{\subsection{Proof of Fact \ref{Mrtilde}} \label{ProofMrTilde}
We recall a fact from \cite[page~134]{bernstein2009matrix} and restate it here.  
\begin{fact}\label{triplematrixproduct}
Let $E \in \mathbb{R}^{a \times b}$, $F \in \mathbb{R}^{b \times c}$ and $G \in \mathbb{R}^{c \times d}$ and $a,b,c,d \in \mathbb{N}$. Then, $\operatorname{rank}(EF)+\operatorname{rank}(FG)-\operatorname{rank}(F)\leq \operatorname{EFG}.$
\end{fact}
The proof of Fact \ref{Mrtilde} may then be given as follows.
\begin{proof}
Define a permutation matrix \cite[pp.~165-166]{bernstein2009matrix}
\begin{equation}\label{permmatrix}
P_{q,r} \triangleq \begin{bmatrix}0 &0 &0 &\cdots &I_q\\
0 &0 &0 &\iddots &0\\
0 &0 &I_q &\cdots &0\\
0 &I_q &0 &\cdots &0\\
I_q &0 &0 &\cdots &0
 \end{bmatrix} \in \mathbb{R}^{(r-L+1)q \times (r-L+1)q}.
\end{equation}
Note that, $P_{q,r}^2=I_{(r-L+1)q}$ and $\operatorname{rank}(P_{q,r})=(r-L+1)q$. It is noted that pre-multiplying or post-multiplying a given matrix by a permutation matrix amounts to permutation/interchange  of rows or columns respectively \cite[pp.~21-22]{banerjee2014linear}, \cite[pp.~61-62]{gentle2007matrix} and therefore does not affect the number of linearly dependent rows or columns in the resultant matrix and therefore does not affect the rank of the resultant matrix. Next, noting that $\tilde{M}_r = P_{l,r} M_r P_{m,r}$ and using Fact \ref{triplematrixproduct} with $E=P_{l,r}$, $G=M_r$, $G=P_{l,r}$ we have, $$\operatorname{rank}(P_{l,r}M_r)+\operatorname{rank}(M_rP_{m,r})-\operatorname{rank}(M_r)\leq \operatorname{rank}(P_{l,r} M_r P_{m,r}).$$
Noting that, $\operatorname{rank}(P_{l,r}M_r)=\operatorname{rank}(M_rP_{m,r})=\operatorname{rank}(M_r)$, we now have
$$\operatorname{rank}(M_r)\leq \operatorname{rank}(P_{l,r}) M_r P_{m,r}),$$ and therefore
$\operatorname{rank}(M_r)= \operatorname{rank}( \tilde{M}_r)$. \end{proof}}}

\end{document}